\documentclass[11pt]{article}
\usepackage[paper=a4paper, margin=2.5cm]{geometry}

\usepackage{multirow}
\usepackage{amsmath}
\usepackage{amssymb}
\usepackage{amsthm}
\usepackage{ifpdf}
\usepackage{authblk}
\usepackage[%
  ocgcolorlinks,%
  linkcolor=blue,%
  filecolor=blue,%
  citecolor=blue,%
  urlcolor=blue]{hyperref}
\usepackage[singlelinecheck=false, labelfont=bf]{caption}

\newtheorem{theorem}{Theorem}[section]
\newtheorem{lemma}[theorem]{Lemma}
\newtheorem{observation}[theorem]{Observation}
\newtheorem{corollary}[theorem]{Corollary}

\newtheorem{claim}[theorem]{Claim}

\theoremstyle{definition}

\newtheorem{remark}[theorem]{Remark}

\newtheorem*{definition*}{Definition}
\theoremstyle{plain}

\newenvironment{DenseItemize}[0]
{\begin{itemize}}
{\end{itemize}}

\newenvironment{MathMaybe}[0]
{\begin{displaymath}\ignorespaces}
{\end{displaymath}\ignorespacesafterend}
\newcommand{\Left}{\left}
\newcommand{\Right}{\right}

\newenvironment{subproof}[1][\proofname]{%
\begin{proof}[#1]%
}{%
\end{proof}%
}

\def\qedlabel#1{\def\theqedlabel{#1}}

\newcommand{\InlineBreakForOverFull}{}

\begin{document}

\title{Bayesian Generalized Network Design\footnote{An extended abstract of this paper is to appear in the 27th Annual European Symposium on Algorithms (ESA 2019).}}

\author{Yuval Emek\footnote{The work of Yuval Emek was supported in part by an Israeli Science Foundation grant number 1016/17.} , Shay Kutten\footnote{The work of Shay Kutten was supported in part by a grant from the ministry of science in the program that is joint with JSPS and in part by the BSF.} , Ron Lavi\footnote{The work of Ron Lavi was partially supported by the ISF-NSFC joint research program (grant No. 2560/17).} , Yangguang Shi\footnote{The work of Yangguang Shi was partially supported at the Technion by a fellowship of the Israel Council for Higher Education.
}}
\affil{Technion - Israel Institute of Technology.\\
\texttt{\{yemek, kutten, ronlavi, shiyangguang\}@ie.technion.ac.il}}

\date{}

\maketitle

\begin{abstract}
We study network coordination problems, as captured by the setting of
\emph{generalized network design} (Emek et al., STOC 2018), in the face of
uncertainty resulting from partial information that the network users hold
regarding the actions of their peers.
This uncertainty is formalized using Alon et al.'s \emph{Bayesian ignorance}
framework (TCS 2012).
While the approach of Alon et al.\ is purely combinatorial, the current paper
takes into account computational considerations:
Our main technical contribution is the development of (strongly) polynomial
time algorithms for local decision making in the face of Bayesian
uncertainty.
\end{abstract}


\vskip 5mm




\section{Introduction}
\label{section:introduction}
In real-life situations, network users are often required to coordinate
actions for performance optimization.
This challenging coordination task becomes even harder in the face of
\emph{uncertainty}, as users often act with partial information regarding
their peers.
Can users overcome their \emph{local} views and reach a good \emph{global}
outcome?
How far would this outcome be from optimal?

For a formal treatment of the aforementioned questions, we adopt the
\emph{Bayesian ignorance} framework of Alon et al.~\cite{Alon2012BI}.
Consider $N$ \emph{agents} in a \emph{routing} scenario, where each agent
$i \in [N]$
should decide on a
$(u_{i}, v_{i})$-path
$a_{i}$ in the network with the objective of minimizing some global \emph{cost
function} that depends on the links' load.
The
$(u_{i}, v_{i})$
pair, also referred to as the \emph{type} of agent $i$, is drawn from a
distribution $p_{i}$.
All agents know this distribution, but the actual realization
$(u_{i}, v_{i})$
of each agent $i$ is only known to $i$ herself.

Our goal is to construct a \emph{strategy} for each agent $i$ that determines
her action $a_{i}$ based only on her individual type $(u_{i}, v_{i})$.
These strategies are computed in a ``preprocessing stage'' and the actual
decision making happens in real-time without further communication.
We measure the quality of a tuple of strategies in terms of its \emph{Bayesian
competitive ratio (BCR)} defined as the ratio of the expected cost obtained by
these strategies to that of an optimal solution computed by an omnipotent
algorithm (refer to Section~\ref{section:bayesian-routing} for the exact
definition).
To the best of our knowledge, this algorithmic evaluation measure has not
been studied so far.

Our main technical contribution is a generic framework that yields strongly
polynomial-time algorithms constructing agent strategies with low BCR for
\emph{Bayesian generalized network design (BGND)} problems --- a setting that
includes routing and many other network coordination problems.
Our framework assumes cost functions that exhibit \emph{diseconomy of scale
(DoS)} \cite{Andrews2012RPM, Bampis2014EES, Makarychev2014SOP}, capturing the
power consumption of network devices that employ the popular \emph{speed
scaling} technique.

\subsection{Model}
\label{section:model}
For clarity of the exposition, we start with the special case of Bayesian routing
in Section~\ref{section:bayesian-routing} and then present the more general
BGND setting in Section~\ref{section:bgnd}. Conceptually, the new algorithmic problem
of Bayesian routing that we define here is related to oblivious routing~\cite{gupta2006oblivious, Englert2009ORL, Shi2015ROI},
where routing requests should be performed without any knowledge about actual network
traffic. This means that the routing path chosen for a routing request may only depend on the
network structure and the other parameters of the problem. Oblivious algorithms are attractive
as they can be implemented very efficiently in a distributed environment as they base routing
decisions only on local knowledge. As will become formally clear below, Bayesian routing has
a similar flavor, but with an important additional ingredient. We will assume that the algorithm
is equipped with statistical (``Bayesian'') knowledge about network traffic. Thus, in a sense,
we replace internal randomization techniques, that oblivious routing usually employs, with
actual data, while still being oblivious to other actual routing decisions and thus still maintaining
the locality principle.\footnote{This is different from stochastic network design as these algorithms
are not oblivious. More details are given below.}

\subsubsection{Special Case: Bayesian Routing}
\label{section:bayesian-routing}
In the \emph{full information} variant of the \emph{routing} problem, we are
given a (directed or undirected) graph
$G = (V, E)$
and a set of $N$ agents, where each agent
$i \in [N]$
is associated with a node pair
$(u_{i}, v_{i}) \in V \times V$,
referred to as the (routing) \emph{request} of agent $i$.
This request should be satisfied by choosing some
$(u_{i}, v_{i})$-path
in $G$, referred to as the (feasible) \emph{action} of agent $i$, and the
collection of all such paths is denoted by $A_{i}$.

Let
$A = A_{1} \times \cdots \times A_{N}$
be the collection of all \emph{action profiles}.
The \emph{load} on edge
$e \in E$
with respect to action profile
$a \in A$,
denoted by $l_{e}^{a}$, is defined to be the number of agents whose actions
include $e$, that is,
$l_{e}^{a} = |\lbrace i \in [N] : e \in a_{i} \rbrace|$.
The cost incurred by load
$l_{e}^{a}$
on edge $e$ is determined by an (edge specific) \emph{superadditive} cost
function
$F_{e}: \mathbb{R}_{\geq 0} \mapsto \mathbb{R}_{\geq 0}$
such that for any $l \geq 0$,
\begin{equation}
F_{e}(l) \; = \; \xi_{e} \cdot l^{\alpha} \, ,
\label{formula_simple_edge_cost_function}
\end{equation}
where
$\xi_{e} > 0$
(a.k.a.\ the \emph{speed scaling factor}) is a parameter of edge $e$ and
$\alpha > 1$
(a.k.a.\ the \emph{load exponent}) is a global constant parameter.
Such a superadditive cost function captures, for example, the power
consumption of network devices employing the popular \emph{speed scaling}
technique
\cite{YaoDS1995, IraniP2005, BansalKP2007, Nedevschi2008RNE,
Christensen2010IEEE, Albers2010}
that allows the device to adapt its power level to its actual load.
In particular, for those network devices that employ the speed scaling
technique, the value of $\alpha$ generally satisfies
$1 < \alpha \leq 3$
\cite{Intel2004EIS, Wierman2009PAS}.
Another application of the cost function
\eqref{formula_simple_edge_cost_function} with
$\alpha = 2$
is to model the queuing delay of users in a TCP/IP communication networks
\cite{Englert2009oblivious}.
The goal in the (full information) routing problem is to construct an action
profile
$a \in A$
with the objective of minimizing the total cost
$C(a) = \sum_{e \in E} F_{e}(l_{e}^{a})$.

\paragraph{Extending to Partial Information}
In the current paper, we extend the full information routing problem to the
\emph{Bayesian routing} problem, where the request of agent
$i \in [N]$
is not fully known to all other agents.
In this problem variant, agent
$i \in [N]$
is associated with a set $T_{i}$ of \emph{types} so that each type
$t_{i} \in T_{i}$
specifies its own routing request
$(u_{i}^{t_{i}}, v_{i}^{t_{i}}) \in V \times V$.
Let $A_{i}^{t_{i}}$ be the set of all (feasible) actions for (the request of)
type $t_{i}$, namely, all
$(u_{i}^{t_{i}}, v_{i}^{t_{i}})$-paths
in $G$ and let
$A_{i} = \bigcup_{t_{i} \in T_{i}} A_{i}^{t_{i}}$.

Agent $i$ is also associated with a \emph{prior distribution} $p_{i}$
over the types in $T_{i}$ and the crux of the Bayesian routing problem is that
agent $i$ should decide on her action while knowing the realization of her own
prior distribution $p_{i}$ (that is, the routing request she should satisfy)
but without knowing the realizations of the prior distributions of the other
agents
$j \neq i$.
Formally, let
$T = T_{1} \times \cdots \times T_{N}$
be the collection of \emph{type profiles}
and
$A = A_{1} \times \cdots \times A_{N}$
be the collection of \emph{action profiles}.
The set of (feasible) action profiles for a type profile
$t \in T$
is denoted by
$A^{t} = A_{1}^{t_{1}} \times \cdots \times A_{N}^{t_{N}}$
and the prior distribution over the type profiles in $T$ is denoted by
$p$.
In this paper, $p$ is assumed to be a product distribution, i.e., the
probability of type profile
$t \in T$
is
$p(t) = \prod_{i = 1}^{N} p_{i}(t_{i})$.

The goal in the Bayesian routing problem is to construct for each agent
$i \in [N]$,
a \emph{strategy}
$s_{i} : T_{i} \mapsto A_{i}$
that maps agent $i$'s realized type
$t_{i} \in T_{i}$
to an action
$a_{i} \in A_{i}^{t_{i}}$.
We emphasize that the decision of agent $i$ is taken irrespective of the other
agents' realized types which are not (fully) known to agent $i$.
Intuitively, a strategy $s_{i}$ can be viewed as a \emph{lookup table}
constructed in the ``preprocessing stage'', and queried at real-time to
determine a (fixed) path for every
$(u_{i}, v_{i})$
pair associated with $i$ (cf.\ \emph{oblivious routing}
\cite{Harald2009SOR, Shi2015ROI}).

The set of strategies available for agent $i$ is denoted by $S_{i}$ and
$S = S_{1} \times \cdots \times S_{N}$
denotes the set of \emph{strategy profiles}.
For each type profile
$t \in T$,
the strategy profile
$s \in S$
determines an action profile
$a = s(t) \in A$
defined so that
$a_{i} = s_{i}(t_{i})$,
$i \in [N]$.
Using this notation, the objective in the Bayesian routing problem is to
construct a strategy profile
$s \in S$
that minimizes the total cost
\[
C(s)
\; = \;
\mathbb{E}_{t \sim p}
\left[ \sum_{e \in E} F_{e} \left( l_{e}^{s(t)} \right)
\right] \, .
\]

\paragraph{Bayesian Competitive Ratio}
Consider an algorithm $\mathcal{A}$ that given a Bayesian routing instance,
constructs a strategy profile $s$.
To evaluate the performance of $\mathcal{A}$, we compare the total cost $C(s)$
to
$\mathbb{E}_{t \sim p}[\mathtt{OPT}(t)]$,
where
\[
\mathtt{OPT}(t)
\; = \;
\min_{a \in A^{t}} \sum_{e \in E} F_{e} (l_{e}^{a})
\]
is the cost of an optimal action profile for the type profile
$t \in T$.
This can be regarded as the expectation, over the same prior distribution $p$,
of the total cost incurred by an omnipotent algorithm that has a global view
of the whole type profile $t$ and enjoys unlimited computational resources.
The \emph{Bayesian competitive ratio (BCR)} of algorithm $\mathcal{A}$ is the
smallest
$\beta \geq 1$
such that for every Bayesian routing instance, the strategy profile $s$
constructed by $\mathcal{A}$ satisfies
$C(s) \leq \beta \cdot \mathbb{E}_{t \sim p}[\mathtt{OPT}(t)]$.

\sloppy
Alon et al.~\cite{Alon2012BI} introduced the related criterion of
\emph{Bayesian ignorance} defined as
$\frac{C(s^{*})}{\mathbb{E}_{t \sim p}[\mathtt{OPT}(t)]}$,
where
$s^{*} = \operatorname{argmin}_{s \in S} C(s)$
is an optimal strategy profile for the given instance.
This criterion quantifies the implication of the agents' partial knowledge
regarding the global system configuration, irrespective of the computational
complexity of constructing this optimal strategy profile.
By definition, for any strategy profile
$s \in S$,
\[
C(s)
\; = \;
\mathbb{E}_{t \sim p}\Big[ \sum_{e \in E} F_{e}\big(l_{e}^{s(t)}\big) \Big]
\; \geq \;
\mathbb{E}_{t \sim p}\Big[ \min_{a \in A^{t}}\sum_{e \in E}F_{e}(l_{e}^{a}) \Big]
\]
which implies that the Bayesian ignorance is at least $1$.
Notice that the BCR is equivalent to the product of the
\emph{approximation ratio}
$\frac{C(s)}{C(s^{*})}$
and the Bayesian ignorance,
therefore it evaluates the loss caused by both algorithmic
(computational complexity) considerations and the absence of the global
information.
The first contribution of the current paper is cast in the following theorem.
\par\fussy

\begin{theorem} \label{theorem_simple_result}
For the Bayesian routing problem, there exists an algorithm whose BCR depends
only on the load exponent parameter $\alpha$.
This algorithm is fully combinatorial and runs in strongly polynomial time.
\end{theorem}

We emphasize that the BCR of the algorithm promised in
Theorem~\ref{theorem_simple_result} is independent of the number of agents
$N$, the underlying graph $G$, the speed
scaling factors $\xi_{e}$,
$e \in E$,
and the probability distribution $p$.
Therefore, as $\alpha$ is assumed to be a constant, so is the BCR.

\subsubsection{Bayesian Generalized Network Design}
\label{section:bgnd}

\paragraph{Generalized Network Design}
The (full information) routing problem has recently been generalized by
Emek et al.~\cite{Emek2018AGN} to the wider family of \emph{generalized
network design (GND)} problems.
In its full information form (the form considered in \cite{Emek2018AGN}), a
GND instance is defined over $N$ agents and a set $E$ of \emph{resources}.
Each agent
$i \in [N]$
is associated with an abstract (not necessarily routing) request characterized
by a set
$A_{i} \subseteq 2^{E}$
of (feasible) actions out of which , some action
$a_{i} \in A_{i}$
should be selected.
As in the routing case, the action profile
$a = (a_{1}, \dots, a_{N})$
induces a load of
$l_{e}^{a} = |\lbrace i \in [N] : e \in a_{i} \rbrace|$
on each resource
$e \in E$
that subsequently incurs a cost of
$F_{e}(l_{e}^{a})$,
where
$F_{e}: \mathbb{R}_{\geq 0} \mapsto \mathbb{R}_{\geq 0}$
is a resource specific cost function.
The goal is to construct an action profile
$a \in A = A_{1} \times \cdots \times A_{N}$
with the objective of minimizing
the total cost
$C(a) = \sum_{e \in E} F_{e}(l_{e}^{a})$.

The request of agent 
$i \in [N]$
is said to be \emph{succinctly represented} \cite{Emek2018AGN} if its
corresponding action set $A_{i}$ can be encoded using
$\operatorname{poly}(|E|)$ bits.
Identifying the resource set $E$ with the edge set of an underlying graph
$G$, the routing requests defined in Section~\ref{section:bayesian-routing}
are clearly succinctly represented since each $A_{i}$ corresponds to the set
of $(u_{i}, v_{i})$-paths in $G$, hence $A_{i}$ can be encoded by specifying
$u_{i}$ and $v_{i}$ (and $G$).
Other examples for succinctly represented requests, where the
resource set $E$ is identified with the edge set of an underlying (directed or
undirected) graph
$G = (V, E)$, include:
\begin{DenseItemize}

\item
\emph{multi-routing} requests in directed or undirected graphs, where given a
collection
$D_{i} \subseteq V \times V$
of \emph{terminal} pairs, the action set $A_{i}$ consists of all edge subsets
$F \subseteq E$
such that the subgraph
$(V, F)$
admits a
$(u, v)$-path
for every
$(u, v) \in D_{i}$;
and

\item
\emph{set connectivity} (resp., \emph{set strong connectivity}) in undirected
(resp., directed) graphs, where given a set
$T_{i} \subseteq V$
of \emph{terminals}, the action set $A_{i}$ consists of all edge subsets
$F \subseteq E$
that induce on $G$ a connected (resp., strongly connected) subgraph that spans
$T_{i}$.

\end{DenseItemize}
All requests mentioned (implicitly or explicitly) hereafter are assumed to be
succinctly represented.

\paragraph{Bayesian GND}
In the current paper, we extend the (full information) GND setting to
\emph{Bayesian GND (BGND)}.
This extension is analogous to the extension of full information routing to
Bayesian routing as defined in Section~\ref{section:bayesian-routing}.
In particular, agent
$i \in [N]$
is now associated with a set $T_{i}$ of types, where each type
$t_{i} \in T_{i}$
corresponds to a request whose action set is denoted by
$A_{i}^{t_{i}}$, and a prior distribution $p_{i}$ over the types in $T_{i}$.
A strategy $s_{i}$ of agent $i$ is a function that maps the agent's
realized type
$t_{i} \in T_{i}$
to an action
$s_{i}(t_{i}) \in A_{i}^{t_{i}}$.

Similarly to the notation introduced in Section~\ref{section:bayesian-routing},
let
$T = T_{1} \times \cdots \times T_{N}$
be the set of type profiles.
Let
$A_{i} = \bigcup_{t_{i} \in T_{i}} A_{i}^{t_{i}}$
and let
$A = A_{1} \times \cdots \times A_{N}$
be the set of action profiles.
Let $S_{i}$ be the set of strategies available for agent $i$ and let
$S = S_{1} \times \cdots \times S_{N}$
be the set of strategy profiles.
Given a strategy profile
$s \in S$
and a type profile
$t \in T$,
let
$a = s(t) \in A$
be the action profile defined so that
$a_{i} = s_{i}(t_{i})$,
$i \in [N]$.
The goal in the BGND problem is to construct a strategy profile
$s \in S$
with the objective of minimizing the total cost
\begin{equation} \label{formula_social_cost_function_in_BGND}
C(s)
\; = \;
\mathbb{E}_{t \sim p}
\left[ \sum_{e \in E} F_{e} \left( l_{e}^{s(t)} \right)
\right] \, .
\end{equation}
The BCR of Algorithm $\mathcal{A}$ is the smallest
$\beta \geq 1$
such that for every BGND instance, the strategy profile
$s \in S$
constructed by $\mathcal{A}$ satisfies
$C(s) \leq \beta \cdot \mathbb{E}_{t \sim p}[\mathtt{OPT}(t)]$,
where
\[
\mathtt{OPT}(t)
\; = \;
\min_{a \in A^{t}} \sum_{e \in E} F_{e} (l_{e}^{a}) \, .
\]

\paragraph{Generalized Cost Functions}
In addition to the generalization of (full information) routing to GND,
\cite{Emek2018AGN} also generalizes the cost functions defined in
Eq.~\eqref{formula_simple_edge_cost_function} to cost functions of the form
\begin{equation}
F_{e}(l)
\; = \;
\sum_{j \in [q]} \xi_{e, j} \cdot l^{\alpha_{j}} \, ,
\label{formula_complex_edge_cost_function}
\end{equation}
where
$q$ is a positive integer,
$\xi_{e, j}$
is a positive real for every
$e \in E$
and
$j \in [q]$,
and
$\alpha_{j}$ is a constant real no smaller than $1$ for every
$j \in [q]$.%
\footnote{The cost functions considered in \cite{Emek2018AGN} have a fixed
additional term, capturing the resource's \emph{startup cost}, that makes them
even more general.
Due to technical difficulties, in the current paper we were not able to cope
with this additional term.}
We define
$\alpha_{\max} = \max_{j \in [q]} \alpha_{j}$
and assume hereafter that
$\alpha_{\max} > 1$.
As discussed in \cite{Emek2018AGN}, this generalization of
Eq.~\eqref{formula_complex_edge_cost_function} is not only interesting from a
theoretical perspective, but also makes the model more applicable to practical
network energy saving applications.
Indeed, in realistic communication networks, a link often consists of several
different devices (e.g., transmitter/receiver, amplifier, adapter), all of
which are operating when the link is in use.
As their energy consumption can vary in terms of the load exponents and speed
scaling factors \cite{Wierman2009PAS},
Eq.~\eqref{formula_complex_edge_cost_function} may often provide a more
accurate abstraction of the actual link's power consumption.

\paragraph{Action Oracles}
For a BGND problem $\mathcal{P}$, this paper develops a framework which
generates an algorithm with BCR
$O (\varrho^{\alpha_{\max}})$
when provided with an \emph{action $\varrho$-oracle} for $\mathcal{P}$.
An action $\varrho$-oracle with parameter
$\varrho \geq 1$
for BGND problem $\mathcal{P}$ (cf.\ the \emph{reply $\varrho$-oracles} of
\cite{Emek2018AGN}) is a procedure that given agent
$i \in [N]$,
type
$t_{i} \in T_{i}$,
and a \emph{weight} vector
$w \in \mathbb{R}_{\geq 0}^{E}$,
generates an action
$a_{i} \in A_{i}^{t_{i}}$
such that
$\sum_{e \in a_{i}} w(e) \leq \varrho \cdot \sum_{e \in a'_{i}} w(e)$
for any action
$a'_{i} \in A_{i}^{t_{i}}$.
An \emph{exact action oracle} is an action $\varrho$-oracle with parameter
$\varrho = 1$.

Notice that the optimization problem behind the action oracle is \emph{not} a
BGND problem:
It deals with a \emph{single} type of a \emph{single} agent and the role of
the resource cost functions is now taken by the weight vector.
These differences often make it possible to implement the action oracle with
known (approximation) algorithms.

For example, the Bayesian routing problem, which requires paths between the
given node pairs, admit an exact action oracle implemented using, e.g.,
Dijkstra's shortest path algorithm
\cite{Dijkstra1959PCG, FredmanT1987FHT}.
In contrast, the BGND problem with set connectivity requests in undirected
graphs (P1),
the BGND problem with set strong connectivity requests in directed graphs (P2),
the BGND problem with multi-routing requests in undirected graphs (P3),
and
the BGND problem with multi-routing requests in directed graphs (P4)
do not admit exact action oracles unless
$\mathrm{P} = \mathrm{NP}$
as these would imply exact (efficient) algorithms for the
\emph{Steiner tree},
\emph{strongly connected Steiner subgraph},
\emph{Steiner forest},
and
\emph{directed Steiner forest}
problems, respectively.
However, employing known approximation algorithms for the latter (Steiner)
problems, one concludes that
BGND problem (P1) admits an action $\varrho$-oracle for
$\varrho \leq 1.39$
\cite{Byrka2013STA};
BGND problem (P2) admits an action $\nu^{\epsilon}$-oracle, where $\nu$ is the
number of terminals \cite{Charikar1998AAD};
BGND problem (P3) admits an action $2$-oracle \cite{Agrawal1995TCA};
and
BGND problem (P4) admits an action
$k^{1 / 2 + \epsilon}$-oracle,
where $k$ is the number of terminal pairs \cite{Chekuri2011SCP}.
This means, in particular, that BGND problems (P1) and (P3) always admit an
action $\varrho$-oracle with a constant approximation ratio $\varrho$, whereas
BGND problems (P2) and (P4) admit such an oracle when $\nu$ and $k$ are
fixed
\cite{Agrawal1995TCA, Charikar1998AAD, Chekuri2011SCP, Byrka2013STA}.
The guarantees of our approximation framework are cast in the following
theorem.

\begin{theorem} \label{theorem_main_result}
Consider a BGND problem $\mathcal{P}$ with an action $\varrho$-oracle
$\mathcal{O}_{\mathcal{P}}$.
When provided access to $\mathcal{O}_{\mathcal{P}}$, the framework proposed in
this paper generates an algorithm $\mathcal{A}_{\mathcal{P}}$ whose BCR
depends only on the load exponent parameters
$\alpha_{1}, \dots, \alpha_{q}$
of Eq.~\eqref{formula_complex_edge_cost_function}.
This framework is fully combinatorial and runs in strongly polynomial time,
hence if $\mathcal{O}_{\mathcal{P}}$ can be implemented to run in strongly
polynomial time, then so can $\mathcal{A}_{\mathcal{P}}$.
\end{theorem}

Again, we emphasize that the BCR of the algorithm promised in
Theorem~\ref{theorem_main_result} is independent of the number of agents
$N$, the number of resources $|E|$, the speed
scaling factors $\xi_{e, j}$,
$j \in [q]$,
$e \in E$,
and the probability distribution $p$.
Therefore, as
$\alpha_{1}, \dots, \alpha_{q}$
are assumed to be constants, so is the BCR.
Since the Bayesian routing problem admits an exact action oracle,
Theorem~\ref{theorem_simple_result} follows trivially from
Theorem~\ref{theorem_main_result}.
Throughout the remainder of this paper, we focus on the BGND framework
promised in Theorem~\ref{theorem_main_result}.

\subsection{Related Works}
\label{section:related-work}
The technical framework that we use is inspired by~\cite{Emek2018AGN}.
Section~\ref{section:challenges-techniques} gives a detailed technical overview
including a full comparison.

In the full information case, network design problems with superadditive cost
functions as defined in Eq.~\eqref{formula_simple_edge_cost_function} have
been extensively studied with the motivation of improving the energy
efficiency of networks
\cite{Andrews2012RPM, Bampis2014EES, Makarychev2014SOP}.
To the best of our knowledge, none of these studies has been extended to the
Bayesian case.

In the research works on oblivious routing (e.g., \cite{Englert2009ORL, Shi2015ROI, Lawler2009MTE, Harsha2008MAL}), the absence of global information in
routing is modeled in an adversarial
(non-Bayesian) manner.
In particular, oblivious routing assumes that no knowledge
about $t_{-i}$ is available when determining every $a_{i}$, and the
performance of the algorithm is evaluated by means of its \emph{competitive
ratio}
$\max_{t \in T} \frac{\sum_{e \in E}F_{e}(l_{e}^{s(t)})}{\mathtt{OPT}(t)}$. 
For the cost function $F_{e}(l) = l^{\alpha}$ with $\alpha > 1$, Englert and R\"{a}cke \cite{Englert2009ORL} propose an $O(\log^{\alpha} |V|)$-competitive oblivious routing algorithm for the scenario where the traffic requests are allowed to be partitioned into fractional flows. 
Shi et al. \cite{Shi2015ROI} prove that for such a cost function,
there exists no oblivious routing algorithm with competitive ratio
$O \left( |E|^{\frac{\alpha - 1}{\alpha + 1}} \right)$ when it is required to choose an integral path for every request.

The Bayesian approach is often used in the game theoretic literature to model
the uncertainty a player experiences regarding the actions taken by the other
players.
Roughgarden \cite{Roughgarden2012POAIC} studies a \emph{routing game} (among
other things) in which the players share (equally) the cost of the edges they
use and proposes a theoretical tool called \emph{smoothness} to analyze the
\emph{price of anarchy (PoA)} of this game in a Bayesian setting, defined as
$\frac{\max_{s \in S^{\operatorname{BNE}}} C(s)}
{\mathbb{E}_{t \sim p}[\mathtt{OPT}(t)]}$,
where $S^{\operatorname{BNE}}$ denotes the set of \emph{Bayes-Nash
equilibria}.
In particular, he proves that with the cost function
$F_{e}(l) = \xi_{e, 1} \cdot l + \xi_{e, 2} \cdot l^{2}$,
the PoA is bounded by $\frac{5}{2}$.
We employ the smoothness toolbox in our algorithmic construction, as further
described in Section~\ref{section:smoothness} (see also the overview in
Section~\ref{section:challenges-techniques}).

Alon et al.~\cite{Alon2012BI} investigate the Bayesian routing game with a
constant cost function
$F_{e} = \xi_{e}$
and prove that the Bayesian ignorance
$\frac{C(s^{*})}{\mathbb{E}_{t \sim p}[ \mathtt{OPT}(t)]}$
is bounded by $O(N)$ (resp.,
$O(\log |E|)$)
in directed (resp., undirected) graphs $G = (V, E)$.
They also introduce game theoretic variants of the Bayesian ignorance notion
and analyze them in that game.

To deal with the inherent uncertainty of the demand in realistic networks,
many research works have been conducted on \emph{stochastic network design}
\cite{Gupta2004ETS, Crainic2011PHB, Rahmaniani2018ABD},
formulated as a \emph{two-stage stochastic optimization} problem:
in the first stage, each link in the network has a fixed cost and the
algorithm needs to make decisions to purchase links knowing the probability
distribution over the network demands;
in the second stage, the network demands are realized (according to the
aforementioned probability distribution) and should be satisfied, which may
require purchasing additional links, this time with an inflated cost.
The objective is to minimize the total cost of the two stages plus a load
dependent term, in expectation.

The BGND setting considered in the current paper is different from two-stage
stochastic optimization (particularly, stochastic network design) in several
aspects, the most significant one is that in BGND, an agent's strategy should
dictate her ``complete action'' (e.g., a path for routing requests) for every
possible type, obliviously of the realized types of the other agents.
In particular, one cannot ``update'' the agents' actions and purchase
additional resources at a later stage to satisfy the realized demands.
Moreover, the current paper evaluates the performance of a BGND algorithm by
means of its BCR that takes into consideration computational complexity
limitations as well as the lack of global information (see
Section~\ref{section:model}) whereas the literature on two-stage stochastic
optimization typically evaluates algorithms using standard approximation
guarantees that accounts only for computational complexity limitations.

In \cite{Garg2008SAO}, Garg et al.~investigate online combinatorial optimization problems where the requests arriving online are drawn independently and identically from a known distribution. As an example, Garg et al.~\cite{Garg2008SAO} study the online Steiner tree problem on an undirected graph $G = (V, E)$. In this problem, at each step the algorithm receives a terminal that is drawn independently from a distribution over $V$, and needs to maintain a subset of edges connecting all the terminals received so far.

Our work differs from \cite{Garg2008SAO} in following four aspects. 
First, 
in the stochastic online optimization problem studied in \cite{Garg2008SAO}, when each request $i$ arrives, the previous requests $\lbrace 1, \cdots, i - 1 \rbrace$ have been realized, and the realization is known. By contrast, in the BGND problem, every agent $i$ needs to be served without knowing the actual realization of the other agents. 
Second, 
the cost function studied in \cite{Garg2008SAO} maps each resource $e$ to a fixed toll, which is subaddtive in the number of requests using $e$, while our cost function is superaddtive. 
Third, 
in the BGND problem with the set connectivity requests, for each agent $i$, each type $t_{i}$ is a set of terminals rather than a single terminal, and each action in $A_{i}^{t_{i}}$ is a Steiner tree spanning over the set of terminals corresponding to $t_{i}$. 
Fourth, 
in the BGND problem, each prior distribution $p_{i}$ is over the types of agent $i$, while there is no distribution over the agents. 

\subsection{Paper Organization}
\label{section:organization}
The rest of this paper is organized as follows.
Section~\ref{section:preliminaries} introduces some of the concepts employed
in our approximation framework together with some notation and terminology.
The main challenges that we had to overcome when developing this framework and
some of the techniques used for that purpose are discussed in
Section~\ref{section:challenges-techniques}.
Section~\ref{section:algorithm} is dedicated to a detailed exposition of our
approximation framework.
Its performance is then analyzed in Section~\ref{section:bound-bcr} using
certain game theoretic properties which are investigated in Sections
\ref{section:smoothness}--\ref{section:efficient-estimation-cost-share}.

\section{Preliminaries}
\label{section:preliminaries}
We follow the common convention that for an $N$-tuple
$x = (x_{1}, \dots, x_{N})$
and for
$i \in [N]$,
the notation $x_{-i}$ denotes the
$(N - 1)$-tuple
$(x_{1}, \dots, x_{i - 1}, x_{i + 1}, \dots, x_{N})$.
Likewise, for a Cartesian product
$X = X_{1} \times \cdots \times X_{N}$
and for
$i \in [N]$,
the notation $X_{-i}$ denotes the Cartesian product
$X_{1} \times \cdots \times X_{i - 1} \times X_{i + 1} \times \cdots \times
X_{N}$.

\subsection{The BGND Game}
\label{section:bgnd-game}
Given an instance
$\mathcal{I} =
\left\langle
N,
E,
\lbrace T_{i}, p_{i} \rbrace_{i \in [N]},
\lbrace \xi_{e, j} \rbrace_{e \in E, j \in [q]},
\lbrace \alpha_{j} \rbrace_{j \in [q]}
\right\rangle$
of a BGND problem $\mathcal{P}$, we define a \emph{BGND game} by associating
every agent
$i \in [N]$
with a strategic player who decides on the strategy $s_{i}$ with the objective
of minimizing her own individual cost defined as follows.
Given an action profile
$a \in A$
and a resource
$e \in E$,
the corresponding cost $F_{e}(l_{e}^{a})$ is equally divided among the players
$i \in [N]$
satisfying
$e \in a_{i}$;
in other words, the \emph{cost share} of player $i$ in resource $e$ under
action profile $a$, denoted by $f_{i, e}(a)$, is defined to be
\[
f_{i, e}(a)
\, = \,
\begin{cases}
0 \, , & e \notin a_{i} \\
\frac{F_{e} \left( l_{e}^{a} \right)}{| i : e \in a_{i} |}
=
\sum_{j} \xi_{e, j} \left( l_{e}^{a} \right)^{\alpha_{j} - 1} \, , & \text{otherwise}
\end{cases}
 \, .
\]
Informally, the individual cost of player $i$ is the sum of her cost shares
over all resources.

For a more formal treatment of the BGND game, we occasionally need to
explicitly specify the type $t_{i}$ of player $i$ in the expressions involving
her cost share in which case we use the notation
$f_{i, e}(t_{i}; a)$,
following the convention that
$f_{i, e}(t_{i}; a) = f_{i, e}(a)$
if
$a_{i} \in A_{i}^{t_{i}}$;
and
$f_{i, e}(t_{i}; a) = \infty$
otherwise.
The individual cost of a player $i$ with respect to the type $t_{i}$ and a
fixed action profile $a$ is defined as
$C_{i}(t_{i}; a) = \sum_{e \in E} f_{i, e}(t_{i}; a)$.
Correspondingly, for each player
$i \in [N]$
and each type
$t_{i} \in T_{i}$,
we define the \emph{type-specified} expected individual cost
\[
C_{i}(t_{i}; s)
\; = \;
\mathbb{E}_{t_{-i} \sim p_{-i}} \left[ C_{i}(t_{i}; s(t_{i}, t_{-i})) \right]
\, .
\]
The objective function that player $i$ wishes to minimize is her
\emph{type-averaged} expected individual cost
\[
C_{i}(s)
\; = \;
\mathbb{E}_{t_{i} \sim p_{i}} \left[ C_{i}(t_{i}; s) \right] \, ,
\]
irrespective of the total cost $C(s)$, often referred to as the \emph{social
cost}.

\begin{observation}
The social cost satisfies
$C(s) = \sum_{i \in [N]} C_{i}(s)$
for every strategy profile
$s \in S$.
\end{observation}

Let
$\mathfrak{f}_{i, e}(a_{i}; s_{-i})
=
\mathbb{E}_{t_{-i} \sim p_{-i}}[f_{i, e}(a_{i}, s_{-i}(t_{-i}))]$
be the expected cost share of player
$i \in [N]$
on resource
$e \in E$
with respect to action
$a_{i} \in A_{i}$
and strategy profile
$s_{-i} \in S_{-i}$.
Fixing
$a_{-i} \in A_{-i}$
(resp.,
$s_{-i} \in S_{-i}$),
the cost share
$f_{i, e}(a_{i}, a_{-i})$
(resp., expected cost share
$\mathfrak{f}_{i, e}(a_{i}; s_{-i})$)
of player $i$ on resource $e$ is the same for every action
$a_{i} \in A_{i}$
such that
$e \in a_{i}$.
Therefore, it is often convenient to ignore the specifics of action $a_{i}$
and use the notations
$f_{i, e}(+, a_{-i})$
and
$\mathfrak{f}_{i, e}(+; s_{-i})$
instead of
$f_{i, e}(a_{i}, a_{-i})$
and
$\mathfrak{f}_{i, e}(a_{i}; s_{-i})$,
respectively,
given that
$e \in a_{i}$.%
\footnote{To avoid ambiguity concerning the definition of
$f_{i, e}(+, a_{-i})$
and
$\mathfrak{f}_{i, e}(+; s_{-i})$
for resources
$e \notin A_{i}$,
we assume (in the scope of using these notations) that
$A_{i} = E$
for all
$i \in [N]$.
This is without loss of generality as one can augment $T_{i}$ with a virtual
type
$\tilde{t}_{i}$
such that
$A_{i}^{\tilde{t}_{i}} = \{ E \}$
and
$p_{i}(\tilde{t}_{i})$
is arbitrarily small.}

\subsection{Definitions for the Algorithm Design and Analysis}
\label{section:definitions}
The following definitions play key roles in the design and analysis of
our approximation framework.

\begin{definition*}[Choice Function \cite{Roughgarden2012POAIC}]
A \emph{choice function}
$\sigma: T \mapsto A$
maps every type profile
$t \in T$
to an action profile
$a \in A^{t}$.
The action specified by $\sigma$ for player
$i \in [N]$
with respect to type profile $t$ is denoted by $\sigma_{i}(t)$.
In particular, the choice function that maps each type profile $t$ to an action
profile that realizes $\mathtt{OPT}(t)$ is denoted by $\sigma^{*}$.
\end{definition*}

\begin{definition*}[Smoothness \cite{Roughgarden2012POAIC}]
Given parameters
$\lambda > 0$
and
$0 < \mu < 1$,
a BGND game is said to be
\emph{$(\lambda, \mu)$-smooth}
if
\[
\sum_{i \in [N]} C_{i}(t_{i}; (\sigma_{i}^{*}(t), a_{-i}))
\; \leq \;
\lambda \cdot \mathtt{OPT}(t) + \mu \cdot \sum_{i \in [N]} C_{i}(t'_{i}, a)
\]
for every type profiles
$t, t' \in T$
and
action profile
$a \in A^{t'}$.
\end{definition*}

\begin{definition*}[Potential Function]
A function
$\Phi: S \mapsto \mathbb{R}_{\geq 0}$
is said to be a \emph{potential function} of the BGND game if
\[
\Phi(s) - \Phi(s_{i}', s_{-i})
\; = \;
C_{i}(s) - C_{i}(s_{i}', s_{-i})
\]
for every
strategy profile
$s \in S$,
player
$i \in [N]$,
and
strategy $s'_{i} \in S_{i}$.
The potential function $\Phi(\cdot)$ is said to be \emph{$K$-bounded} for a
parameter
$K \geq 1$
if
$\Phi(s)
\leq
C(s)
\leq
K \cdot \Phi(s)$
for every strategy profile
$s \in S$.
\end{definition*}

\sloppy
\begin{definition*}[$(\underline{\eta}, \overline{\eta})$-Estimation]
Given real parameters
$\underline{\eta}, \overline{\eta} \geq 1$,
a value $x$ is said to be an
\emph{$(\underline{\eta}, \overline{\eta})$-estimation}
of the expected cost share
$\mathfrak{f}_{i, e}(a_{i}; s_{-i})$
(resp.,
$\mathfrak{f}_{i, e}(+; s_{-i})$)
if it satisfies
$x / \underline{\eta} \; \leq \;
\mathfrak{f}_{i, e}(a_{i}; s_{-i})
\; \leq \;
x \cdot \overline{\eta}$
(resp.,
$x / \underline{\eta} \; \leq \;
\mathfrak{f}_{i, e}(+; s_{-i})
\; \leq \;
x \cdot \overline{\eta}$).
We typically denote this estimation $x$ by
$\widehat{\mathfrak{f}}_{i, e}(a_{i}; s_{-i})$
(resp.,
$\widehat{\mathfrak{f}}_{i, e}(+; s_{-i})$).
The BGND game is said to be \emph{poly-time
$(\underline{\eta}, \overline{\eta})$-estimable}
if for every player
$i \in [N]$
and strategy profile
$s_{-i} \in S_{-i}$,
there exists an algorithm which runs in time
$\operatorname{poly}(N, q, |T_{1}|, \cdots, |T_{N}|)$
and outputs an
$(\underline{\eta}, \overline{\eta})$-estimation of the expected cost share
$\mathfrak{f}_{i, e}(+; s_{-i})$.
The BGND game is said to be \emph{tractable} if it is poly-time
$(\underline{\eta}, \overline{\eta})$-estimable with
$\overline{\eta} = \underline{\eta} = 1$.
\par\fussy

Fix some player
$i \in [N]$,
type
$t_{i} \in T_{i}$,
and
$(\underline{\eta}, \overline{\eta})$-estimations
$\widehat{\mathfrak{f}}_{i, e}(s_{i}(t_{i}); s_{-i})$,
$e \in E$.
With respect to these variables, let
$\widehat{C}_{i}(t_{i}; s)
=
\sum_{e \in E} \widehat{\mathfrak{f}}_{i, e}(s_{i}(t_{i}); s_{-i})$
and
$\widehat{C}_{i}(s)
=
\mathbb{E}_{t_{i} \sim p_{i}}[\widehat{C}_{i}(t_{i}; s)]$.
By the linearity of expectation, we know that
\[
\widehat{C}_{i}(t_{i}; s) / \underline{\eta}
\, \leq \,
C_{i}(t_{i}; s)
\, \leq \,
\widehat{C}_{i}(t_{i}; s) \cdot \overline{\eta}
\quad \text{and} \quad
\widehat{C}_{i}( s) / \underline{\eta}
\, \leq \,
C_{i}(s)
\, \leq \,
\widehat{C}_{i}(s) \cdot \overline{\eta} \, .
\]
Consequently, we refer to
$\widehat{C}_{i}(t_{i}; s)$
and $\widehat{C}_{i}(s)$
as
\emph{$(\underline{\eta}, \overline{\eta})$-estimations}
of
$C_{i}(t_{i}; s)$
and $C_{i}(s)$, respectively.
\end{definition*}

\begin{definition*}[Approximate Best Response]
For strategy profile
$s \in S$
and player
$i \in [N]$,
strategy
$s_{i} \in S_{i}$
is said to be an \emph{approximate best response (ABR)} of $i$ with
approximation parameter
$\chi \geq 1$
if
$C_{i}(s_{i}, s_{-i}) \leq \chi \cdot C_{i}(s_{i}', s_{-i})$
holds for any
$s'_{i} \in S_{i}$.
We may omit the explicit mention of the approximation parameter $\chi$ when it
is clear from the context.
A \emph{best response (BR)} is an ABR with approximation parameter
$\chi = 1$.
\end{definition*}

\begin{definition*}[Approximate Best Response Dynamics]
An \emph{approximate best response dynamic (ABRD)} is a procedure that starts
from a predetermined strategy profile
$s^{0} \in S$
and generates a series of strategy profiles
$s^{1}, \cdots, s^{R}$
such that for every
$1 \leq r \leq R$,
there exists some player
$i \in [N]$
satisfying
(1)
$s_{-i}^{r} = s_{-i}^{r - 1}$;
and
(2)
$s_{i}^{r}$ is an ABR of $i$ to $s_{-i}^{r - 1}$.
\end{definition*}

\section{Overview of the Main Challenges and Techniques}
\label{section:challenges-techniques}
The approximation framework presented in Section \ref{section:algorithm} for BGND problems
is inspired by the framework designed in \cite{Emek2018AGN} for full information
GND problems only in the conceptual sense that both algorithms employ approximate best
response dynamics. In a high-level, for a certain number $R$ of rounds that will be carefully
chosen in order to achieve the approximation promise, and starting from some properly
chosen initial strategy profile $s^{0}$, for each round $1 \leq r \leq R$ the strategy profile
$s^{r}$ is generated from $s^{r-1}$ in the following manner:

\begin{enumerate}

\item \label{item:overview:compute-estimations}
For every player
$i \in [N]$
and resource
$e \in E$,
compute an
$(\underline{\eta}, \overline{\eta})$-estimation
$\widehat{\mathfrak{f}}_{i, e}(+; s_{-i}^{r - 1})$
of the expected cost share
$\mathfrak{f}_{i, e}(+; s_{-i}^{r - 1})$.

\item \label{item:overview:apply-oracle}
For every player
$i \in [N]$,
construct the strategy $s'_{i}$ by mapping each type
$t_{i} \in T_{i}$
to the action
$a_{i} \in A_{i}^{t_{i}}$
computed by invoking the action $\varrho$-oracle with weight vector $w$
defined by setting
$w(e) = \widehat{\mathfrak{f}}_{i, e}(+; s_{-i}^{r - 1})$.

\item \label{item:overview:choose-player}
Choose player
$i \in [N]$
according to the game theoretic criterion presented in
Section~\ref{section:algorithm} regarding the estimations
$\widehat{C}_{i}(s^{r - 1})$
and
$\widehat{C}_{i}(s'_{i}, s_{-i}^{r - 1})$
of the type-averaged expected individual costs.
Construct $s^{r}$ by updating the strategy of the chosen player $i$ to
$s'_{i}$.

\end{enumerate} 

However, beyond the similar high-level structure, the technical construction in this paper is
entirely different from \cite{Emek2018AGN} since the incomplete information assumption
of the BGND setting exhibits new algorithmic challenges that require novel techniques.
Specifically, the main challenges that our technical analysis in this paper handles are as follows.

A first obstacle here is the difficulty in computing the estimation
$\widehat{\mathfrak{f}}_{i, e}(+; s_{-i}^{r - 1})
=
\mathbb{E}_{t_{-i} \sim p_{-i}}[f_{i, e}(+, s_{-i}(t_{-i}))]$
in step~\ref{item:overview:compute-estimations} since there are
exponentially (in $N$) many possibilities for $t_{-i}$.
Another source of difficulty in this regard is that the function
$f_{i, e}(+, s_{-i}(t_{-i}))$
is nonlinear in
$l_{e}^{s_{-i}(t_{-i})}$.
One may hope that Jensen's inequality \cite{Jensen1906SLF} can resolve this
issue, however, as we explain in the technical sections, it is not enough for
obtaining proper bounds on both $\underline{\eta}$ and $\overline{\eta}$.
This obstacle is addressed in
Section~\ref{section:efficient-estimation-cost-share} where we employ
probabilistic tools from \cite{Berend2010IBB} and using Cantelli's inequality
\cite{Savage1961PIT} to obtain the required estimation of the expression
$\mathbb{E}_{t_{-i} \sim p_{-i}}[f_{i, e}(+, s_{-i}(t_{-i}))]$.

A second obstacle is that the ABRD-based approximation framework expresses its approximation guarantees
in terms of smoothness parameters and bounded potential functions. However, neither
the smoothness parameters nor the existence of a bounded potential function are known
for the BGND game that we have defined here. We provide a new analysis for these two
issues in Sections \ref{section:smoothness} and \ref{section:bounded-potential}, respectively.


A third obstacle involves the stopping condition of the best response dynamics. A stopping
condition for the full information case, via the smoothness framework, was developed by
\cite{Roughgarden2015IRP} (showing that if the current outcome in a best response dynamics
is far from optimal there must exist a player whose best response significantly improves his
own utility). For the Bayesian case, to the best of our knowledge, no such general stopping
condition was known prior to the current paper. In fact, the smoothness framework for the
Bayesian case which was developed in \cite{Roughgarden2012POAIC} did not include any
results on best response dynamics. One specific technical difficulty is that Bayesian smoothness
is defined in \cite{Roughgarden2012POAIC} w.r.t.~a deviation to the optimal choice function
rather than to a best response. This obstacle is resolved in Section \ref{section:bound-bcr}
where we provide such a stopping condition by proving that if the outcome of the current
step of the ABRD in the Bayesian case is far from optimal, there must exist a player whose
approximate best response must significantly improve her utility.


A fourth obstacle regards the output of the algorithm, once the ABRD terminates. Although we
prove that there exists at least one strategy profile $s^{r}$, $1 \leq r \leq R$, with a sufficiently
small social cost $C(s^{r})$, we do not know how to find it. In particular, we wish to emphasize
that we cannot simply evaluate the social cost function $C(\cdot)$
(see Eq.~\eqref{formula_social_cost_function_in_BGND}) due to the exponential number of type
profiles. This obstacle does not exist in \cite{Emek2018AGN} where they can explicitly go over all
steps of the full information ABRD and find the exact step whose outcome has minimal cost. To
resolve this issue, we output the last strategy profile $s^{R}$ generated in the ABRD and bound
its loss. This is described in Section~\ref{section:bound-bcr}.

Our technical constructions and our analysis employ various techniques from algorithmic game theory, demonstrating once again (as in~\cite{Emek2018AGN}) the usefulness of this literature as a toolbox for algorithmic constructions that, on the face of it, have nothing to do with selfish agents. In particular, in this paper (and as assumed in the literature on oblivious routing \cite{gupta2006oblivious, Englert2009ORL, Shi2015ROI}), we construct an algorithm that receives a correct input and outputs routing tables that the agents are going to follow without issues of selfish deviations.

\section{The Algorithm}
\label{section:algorithm}
In this part, we present an algorithm, which is referred to as \texttt{Bayes-ABRD}, for a given BGND problem $\mathcal{P}$. The algorithm is assumed to have free access to an action $\varrho$-oracle for $\mathcal{P}$, which is denoted by $\mathcal{O}_{\mathcal{P}}$.

With an input instance
$\mathcal{I} =
\left\langle
N,
E,
\lbrace T_{i}, p_{i} \rbrace_{i \in [N]},
\lbrace \xi_{e, j} \rbrace_{e \in E, j \in [q]},
\lbrace \alpha_{j} \rbrace_{j \in [q]}
\right\rangle$,
the first step of the algorithm is to (conceptually) construct a BGND game,
and choose a tuple of parameters
$(\lambda, \mu, K, \underline{\eta}, \overline{\eta})$
such that the BGND game
\begin{enumerate}
\item is $(\lambda, \mu)$-smooth with $\varrho (\underline{\eta}\overline{\eta})^{2} \mu < 1$, 
\item has a potential function $\Phi$ that is $K$-bounded,
\item is poly-time $(\underline{\eta}, \overline{\eta})$-estimable.
\end{enumerate}
The existence and exact values of the parameters in this tuple are presented in the following sections. In particular, the smoothness parameters $(\lambda, \mu)$ are analyzed in Section \ref{section:smoothness}, the potential function is established in Section \ref{section:bounded-potential}, and the estimation parameters $(\underline{\eta}, \overline{\eta})$ are specified in Section \ref{section:efficient-estimation-cost-share}.

\begin{lemma} \label{lemma:efficient-approximation-bar}
For any $i \in [N]$ and any $s_{-i} \in S_{-i}$,
there exists a $\operatorname{poly}(|E|, N, q, \lbrace |T_{i}| \rbrace_{i \in [N]})$-time procedure which generates a strategy
$s_{i} \in S_{i}$
and the corresponding $(\underline{\eta}, \overline{\eta})$-estimation $\widehat{C}_{i}(s_{i}, s_{-i})$ 
of the individual costs
such that
$\widehat{C}_{i}(s_{i}, s_{-i})
\leq
\varrho \cdot \underline{\eta}\cdot C_{i}(s_{i}', s_{-i})$
for any
$s_{i}' \in S_{i}$.
This means in particular that $s_{i}$ is an ABR
of $i$ to $s_{-i}$ with approximation parameter $\varrho \cdot
\underline{\eta}\overline{\eta}$.\footnote{%
All subsequent occurrences of the term ABR (and
ABRD) share the same approximation parameter $\varrho \underline{\eta}\overline{\eta}$, hence we
may refrain from mentioning this parameter explicitly.
}
\end{lemma}

\begin{proof}
For each player $i \in [N]$, construct the weight vector
$w_{i, s_{-i}} : E \rightarrow \mathbb{R}_{\geq 0}$
by setting $w_{i, s_{-i}}(e)$ to be the $(\underline{\eta}, \overline{\eta})$-estimation $\widehat{\mathfrak{f}}_{i, e}(+;  s_{-i})$ of the expected share. 
This weight vector can be obtained in time $\operatorname{poly}(|E|, N, q, \lbrace |T_{i}| \rbrace_{i \in [N]})$ since the BGND game is poly-time $(\underline{\eta}, \overline{\eta})$-estimable. 
By definition, for any action $a_{i}' \in A_{i}$ satisfying $e \in a_{i}'$, it holds that $\mathfrak{f}_{i, e}(a_{i}'; s_{-i}) = \mathfrak{f}_{i, e}(+; s_{-i})$. 
It implies that $w_{i, s_{-i}}(e)$ can be taken as an $(\underline{\eta}, \overline{\eta})$-estimation $\widehat{\mathfrak{f}}_{i, e}(a_{i}'; s_{-i})$ of the expected share $\mathfrak{f}_{i, e}(a_{i}'; s_{-i})$.

Then, through accessing the action $\varrho$-oracle $\mathcal{O}_{\mathcal{P}}$ for each type $t_{i} \in T_{i}$, a strategy $s_{i}$ can be found such that for any strategy $s_{i}' \in S_{i}$, 
\begin{MathMaybe}
\sum_{e \in E}\widehat{\mathfrak{f}}_{i, e}(s_{i}(t_{i}); s_{-i}) \, = \, \sum_{e \in s_{i}(t_{i})} w_{i, s_{-i}}(e) \, \leq \, \varrho \cdot \sum_{e \in s_{i}'(t_{i})} w_{i, s_{-i}}(e) \, \leq \, \varrho \cdot \underline{\eta} \sum_{e \in E}\mathfrak{f}_{i, e}(s_{i}'(t_{i}); s_{-i}) \, ,
\end{MathMaybe}
which means that $\widehat{C}_{i}(t_{i}; (s_{i}, s_{-i})) \leq \varrho \cdot \underline{\eta} \cdot C_{i}(t_{i}; (s_{i}', s_{-i}))$.

By the linearity of the expectation, $\sum_{t_{i} \in T_{i}}p_{i}(t_{i})\sum_{e \in s_{i}(t_{i})} w_{i, s_{-i}}(e)$ gives the desired $(\underline{\eta}, \overline{\eta})$-estimation $\widehat{C}_{i}(s_{i}, s_{-i})$, and for any $s_{i}' \in S_{i}$, it holds that 
\begin{MathMaybe}
\widehat{C}_{i}(s_{i}, s_{-i}) \; = \; \sum_{t_{i} \in T_{i}}p_{i}(t_{i}) \cdot \widehat{C}_{i}(t_{i}; (s_{i}, s_{-i})) \; \leq \; \varrho \underline{\eta} \sum_{t_{i} \in T_{i}}p_{i}(t_{i}) \cdot C_{i}(t_{i}; (s_{i}', s_{-i})) \; \leq \; \varrho \underline{\eta} C_{i}(s_{i}', s_{-i}) \, .
\end{MathMaybe}
\end{proof}

Employing the procedure promised by
Lemma \ref{lemma:efficient-approximation-bar}, \texttt{Bayes-ABRD} simulates an
ABRD of at most $R$ rounds
$s^{0}, s^{1}, \dots$
for the BGND game induced by $\mathcal{I}$. 
Here $R$ is a positive integer depending on the tuple $(\lambda, \mu, K, \underline{\eta}, \overline{\eta})$, and its exact value is also deferred to the following parts (Section \ref{section:bound-bcr}). 
The ABRD simulated in our algorithm is done as follows.

Each player $i$ chooses her initial strategy $s_{i}^{0}$ by taking each $s_{i}^{0}(t_{i})$ to be the action generated by $\mathcal{O}_{\mathcal{P}}$ for type $t_{i}$ with respect to the weight vector $w^{0}$ defined by setting
$w^{0}(e) = \sum_{j \in [q]}\xi_{e, j}$, that is, as if $i$ is playing alone.
The obtained strategy $s_{i}^{0}$ is broadcast by player $i$ to all the other players such that the full strategy profile $s^{0}$ is known by every player. 
Assuming that $s^{r - 1}$,
$1 \leq r \leq R$,
was already constructed and known by all the players, $s^{r}$ is obtained as follows.
Every player $i \in [N]$ employs the procedure promised by
Lemma \ref{lemma:efficient-approximation-bar} to generate an ABR $\widehat{s}_{i}^{\, r - 1}$ to
$s_{-i}^{r - 1}$, 
and computes
$\Delta_{i}^{r}
=
\widehat{C}_{i}(s^{r - 1}) -
(\underline{\eta}\overline{\eta}) \cdot \widehat{C}_{i}(\widehat{s}_{i}^{\, r - 1}, s_{-i}^{r - 1})$. 
Both the strategy $\widehat{s}_{i}^{\, r - 1}$ and the value $\Delta_{i}^{r}$ are broadcast to all the other players.
If
$\Delta_{i}^{r} \leq 0$
for all
$i \in [N]$,
then the ABRD stops, and every player $i$ sets
$s_{i}^{r} = s_{i}^{r - 1}$;
in this case, we say that the ABRD \emph{converges}.
Otherwise, fix
$\Delta^{r} = \sum_{i \in [N]} \Delta_{i}^{r}$
and choose some player
$i' \in [N]$
so that
\begin{equation}
\label{formula_condition_of_being_selected_for_updating_the_strategy}
\Delta_{i'}^{r} > 0
\quad \text{and} \quad
\Delta_{i'}^{r} \geq \frac{1}{N}\Delta^{r}
\end{equation}
to update her strategy, setting
$s^{r} = (\widehat{s}_{i'}^{\, r - 1}, s_{-i'}^{r - 1})$
(the existence of such a player is guaranteed by the pigeonhole principle, and ties are always broken by choosing the player with the smallest index). Such an update can be performed by each player in a distributed manner, as every player has the knowledge of the full vectors $\lbrace s_{i}^{r} \rbrace_{i \in [N]}$ and $\lbrace \Delta_{i}^{r} \rbrace_{i \in [N]}$.

When the ABRD terminates (either because it has reached round
$r = R$
or because it converges), \texttt{Bayes-ABRD} outputs the strategy generated in the last round.

\begin{remark}
Note that \texttt{Bayes-ABRD} is designed for computing the strategy profile, not for invoking the strategies to decide the actions in real-time. All the operations of \texttt{Bayes-ABRD}, including broadcasting the strategy $\widehat{s}_{i}^{\, r - 1}$ and the value $\Delta_{i}^{r}$ for every player $i$ in every round $r \in [R]$, are carried out in a ``precomputing stage'' without seeing the realized type profile. The decision making that happens in real-time does not involve any further communication.
\end{remark}

\section{Bounding the BCR with Game Theoretic Parameters}
\label{section:bound-bcr}
\begin{lemma}\label{lemma_goodness_of_best_response_wrp_each_specific_type}
For every player $i$ and every strategy profile $s$, if $s_{i}'$ is the BR of $i$ to $s$, then 
\begin{MathMaybe}
C_{i}\big(t_{i}; (s_{i}', s_{-i}) \big) \, \leq \, \mathbb{E}_{t_{-i} \sim p_{-i}}\Big[C_{i}\big(t_{i}; (a_{i}, s_{-i}(t_{-i})) \big) \Big]
\end{MathMaybe}
holds for every type $t_{i}$ and every action $a_{i} \in A_{i}^{t_{i}}$.
\end{lemma}

\begin{proof}
Suppose that there exists a type $t_{i}'$ and an action $a_{i}' \in A_{i}^{t_{i}'}$ such that $C_{i}\big(t_{i}; (s_{i}', s_{-i}) \big) > \mathbb{E}_{t_{-i} \sim p_{-i}}\Big[C_{i}\big(t_{i}; (a_{i}, s_{-i}(t_{-i})) \big) \Big]$. Now construct a new strategy $s_{i}''$ of player $i$ which maps every type $t_{i} \neq t_{i}'$ to the same action as $s_{i}'$, and maps $t_{i}'$ to $a_{i}'$. Then
\begin{align*}
C_{i}(s_{i}'', s_{-i}) \; 
 =& \; \mathbb{E}_{t_{i} \sim p_{i}}\Big[C_{i}\big(t_{i}; (s_{i}'', s_{-i}) \big) \Big] \\
 =& \; \sum_{t_{i} \neq t_{i}'}p_{i}(t_{i})C_{i}\big(t_{i}; (s_{i}'', s_{-i}) \big) + p_{i}(t_{i}')\mathbb{E}_{t_{-i} \sim p_{-i}}\Big[C_{i}\big(t_{i}; (a_{i}, s_{-i}(t_{-i})) \big) \Big] \\
 <& \; \sum_{t_{i} \neq t_{i}'}p_{i}(t_{i})C_{i}\big(t_{i}; (s_{i}'', s_{-i}) \big) + p_{i}(t_{i}')C_{i}(t_{i}', (s_{i}', s_{-i})) \; = \; C_{i}(s_{i}', s_{-i}) \, ,
\end{align*}
which conflicts with the assumption that $s_{i}'$ is the BR of $i$ to $s$.
\end{proof}

\begin{lemma}\label{lemma_no_gain_by_deviating_mixed_strategy_from_pure_strategy}
For a BGND game that is $(\lambda, \mu)$-smooth with $\lambda > 0$ and $0 < \mu < \frac{1}{\varrho(\underline{\eta}\overline{\eta})^{2}}$ and every strategy profile $s$, let $s_{i}'$ be the BR of each player $i$ to $s$, then 
\begin{MathMaybe}
\sum_{i \in [N]}C_{i}(s_{i}', s_{-i}) \; \leq \; \lambda \cdot \mathbb{E}_{t \in T}\left[\mathtt{OPT}(t)\right] + \mu \cdot C(s) \, .
\end{MathMaybe}
\end{lemma}

\begin{proof}
For every fixed $t_{-i}' \in T_{-i}$, Lemma \ref{lemma_goodness_of_best_response_wrp_each_specific_type} indicates that for every $i$, every $t_{i}$, and every $t_{-i}' \in T_{-i}$,
\begin{equation*}
C_{i}(t_{i}; (s_{i}', s_{-i})) \; \leq \; \mathbb{E}_{t_{-i} \sim p_{-i}}\Big[C_{i}\big(t_{i}; (\sigma_{i}^{*}(t_{i}, t_{-i}'), s_{-i}(t_{-i})) \big) \Big] \, ,
\end{equation*}
because the action $\sigma_{i}^{*}(t_{i}, t_{-i}')$ does not depend on $t_{-i}$. Taking the expectation over $t_{i}$, we get
\begin{align*}
\sum_{i \in [N]}C_{i}(s_{i}', s_{-i}) \; =& \; \sum_{i \in [N]}\mathbb{E}_{t_{i} \sim p_{i}}\Big[ C_{i}(t_{i}; (s_{i}', s_{-i})) \Big] \\
\leq& \; \sum_{i \in [N]} \mathbb{E}_{t_{i} \sim p_{i}}\Big[ \mathbb{E}_{t_{-i} \sim p_{-i}}\Big[C_{i}\big(t_{i}; (\sigma_{i}^{*}(t_{i}, t_{-i}'), s_{-i}(t_{-i})) \big) \Big] \Big] \\
=& \; \sum_{i \in [N]} \mathbb{E}_{t \sim p}\Big[ C_{i}\big(t_{i}; (\sigma_{i}^{*}(t_{i}, t_{-i}'), s_{-i}(t_{-i})) \big) \Big] \, .
\end{align*}
The last transition holds because the prior distribution $p$ is assumed to be a product distribution. Since the formula above holds for every $t_{-i}' \in T_{-i}$, it can be derived from the definition of expectation that
\begin{align*}
\sum_{i \in [N]}C_{i}(s_{i}', s_{-i}) \; \leq& \; \mathbb{E}_{t_{-i}' \sim p_{-i}}\bigg[ \sum_{i \in [N]} \mathbb{E}_{t \sim p}\Big[ C_{i}\big(t_{i}; (\sigma_{i}^{*}(t_{i}, t_{-i}'), s_{-i}(t_{-i})) \big) \Big] \bigg] \\
=& \; \sum_{i \in [N]}\mathbb{E}_{t \sim p}\left[ \mathbb{E}_{t_{-i}' \sim p_{-i}} \left[C_{i}\Big(t_{i}; \big(\sigma_{i}^{*}(t_{i}, t_{-i}'), s_{-i}(t_{-i}) \big) \Big) \right] \right] \, .
\end{align*}
The last transition holds because $t_{-i}'$ is independent of $t$. In \cite{Roughgarden2012POAIC}, it is proved that in a BGND game that is $(\lambda, \mu)$-smooth, it holds for any strategy profile $s$ that
\begin{equation*}
\sum_{i \in [N]}\mathbb{E}_{t \sim p}\left[ \mathbb{E}_{t_{-i}' \sim p_{-i}} \left[C_{i}\Big(t_{i}; \big(\sigma_{i}^{*}(t_{i}, t_{-i}'), s_{-i}(t_{-i}) \big) \Big) \right] \right] \; \leq \; \lambda \cdot \mathbb{E}_{t \sim p}\left[\mathtt{OPT}(t)\right] + \mu \cdot C(s) \, .
\end{equation*}
Since $\mu < \frac{1}{\varrho(\underline{\eta}\overline{\eta})^{2}} \leq 1$, this proposition follows.
\end{proof}

\begin{lemma}\label{lemma:ABRD-converge-approx-NE}
If the ABRD simulated in \texttt{Bayes-ABRD} converges at round $r$ for any
$r \in [R]$,
then the last strategy profile $s^{r}$ satisfies
\begin{equation*}
C(s^{r})
\leq
\frac{\varrho (\underline{\eta}\overline{\eta})^{2}\lambda}{1 - \varrho (\underline{\eta}\overline{\eta})^{2} \mu}
\cdot \mathbb{E}_{t \sim T}\Big[\mathtt{OPT}(t)\Big] \, .
\end{equation*}
\end{lemma}

\begin{proof}
Recalling that we use $s_{i}'$ and $\widehat{s}_{i}^{\, r}$ to respectively represent the BR and ABR of player $i$ to
$s^{r}$, we observe that
\begin{align*}
C(s^{r})
\; = \;
\sum_{i}C_{i}(s^{r}) \;
\leq & \;
\overline{\eta} \sum_{i}\widehat{C}_{i}(s^{r}) \\
\leq & \;
\overline{\eta} (\underline{\eta}\overline{\eta}) \sum_{i} \widehat{C}_{i}(\widehat{s}_{i}^{\, r},
s_{-i}^{r}) \\
\leq & \;
\varrho (\underline{\eta}\overline{\eta})^2 \cdot \sum_{i} C_{i}(s_{i}', s_{-i}^{r}) \\
\leq & \;
\varrho (\underline{\eta}\overline{\eta})^2 (\lambda \cdot \mathbb{E}_{t \sim T}\Big[\mathtt{OPT}(t)\Big] + \mu \cdot C(s^{r})) \, ,
\end{align*}
where
the second transitions follow from the definition of the $(\underline{\eta}, \overline{\eta})$-estimation of the individual cost, 
the third transition holds since the ABRD converges at round $r$,
the fourth transition holds following Lemma
\ref{lemma:efficient-approximation-bar},
and
the fifth transition follows from Lemma \ref{lemma_no_gain_by_deviating_mixed_strategy_from_pure_strategy}.
\end{proof}

\begin{lemma}\label{lemma:optimal-reply-guarantees-approx-ratio}
The initial strategy profile $s^{0}$ of \texttt{Bayes-ABRD} satisfies
$C(s^{0}) \leq \varrho \cdot N^{\alpha_{\max} - 1} \cdot \mathbb{E}_{t \sim T}\Big[\mathtt{OPT}(t)\Big]$.
\end{lemma}

\begin{proof}
The construction of $s^{0}$ guarantees that
\[
\sum_{e \in s_{i}^{0}(t_{i})} \sum_{j \in [q]} \xi_{e, j}
\, \leq \,
\varrho \cdot \sum_{e \in \sigma_{i}^{*}(t_{i}, t_{-i})} \sum_{j \in [q]} \xi_{e, j} 
\]
holds for any $i$, any $t_{i}$, and any $t_{-i}$. 
It implies that,
\begin{align*}
\sum_{i \in [N] } \sum_{e \in s_{i}^{0}(t_{i})} \sum_{j \in [q]}\xi_{e, j} 
\, \leq & \,
\varrho \cdot \sum_{i \in [N]} \sum_{e \in \sigma_{i}^{*}(t_{i}, t_{-i})} \sum_{j \in [q]} \xi_{e, j}\\
\leq & \,
\varrho \cdot \sum_{i \in [N]} \sum_{e \in \sigma_{i}^{*}(t_{i}, t_{-i})} \sum_{j \in [q]} \xi_{e, j} \Big(l_{e}^{\sigma^{*}(t_{i}, t_{-i})}\Big)^{\alpha_{j} - 1}\\
= & \,
\varrho \cdot \mathtt{OPT}(t_{i}, t_{-i}) \, ,
\end{align*}
where the second transition holds because $l_{e}^{\sigma^{*}(t)} \in \mathbb{Z}_{\geq 1}$ for any $e \in \sigma_{i}^{*}(t)$, and $\alpha_{j} - 1 \geq 0$.
Then,
\begin{align*}
C(s^{0})
\, =& \,
\mathbb{E}_{t \sim T}\bigg[ \sum_{e \in E} \sum_{j \in [q]} \xi_{e, j} \Big(l_{e}^{s^{0}(t)}\Big)^{\alpha_{j}} \bigg]\\
= & \,
\mathbb{E}_{t \sim T}\bigg[ \sum_{e \in E} \sum_{i: e \in s_{i}^{0}(t_{i})} \sum_{j \in [q]} \xi_{e, j}\Big(l_{e}^{s^{0}(t)}\Big)^{\alpha_{j} - 1} \bigg]\\
\leq & \,
\mathbb{E}_{t \sim T}\Big[ \sum_{e \in E} \sum_{i: e \in s_{i}^{0}(t_{i})} \sum_{j \in [q]} \xi_{e, j}\cdot N^{\alpha_{j} - 1} \Big]\\
\leq & \,
N^{\alpha_{\max} - 1} \mathbb{E}_{t \sim T}\Big[ \sum_{e \in E} \sum_{i: e \in s_{i}^{0}(t_{i})} \sum_{j \in [q]} \xi_{e, j} \Big] \\
= & \,
N^{\alpha_{\max} - 1} \mathbb{E}_{t \sim T}\Big[ \sum_{i \in [N]}\sum_{e \in s_{i}^{0}(t_{i})} \sum_{j \in [q]} \xi_{e, j} \Big]  \\
\leq & \,
\varrho N^{\alpha_{\max} - 1} \cdot \mathbb{E}_{t \sim T}\Big[\mathtt{OPT}(t)\Big] \, .
\end{align*}
The assertion follows.
\end{proof}

\begin{lemma}\label{lemma_potential_function_decreases_in_each_step_of_ABRD}
For any round $r < R$ such that the ABRD does not converge at round $r + 1$, as long as the player selected for strategy update satisfies Eq.~(\ref{formula_condition_of_being_selected_for_updating_the_strategy}), we have $\Phi(s^{r}) - \Phi(s^{r + 1}) > 0$.
\end{lemma}

\begin{proof}
Since the ABRD does not converge at round $r$, there exists a player $i^{r}$ who is selected to update her strategy. By the definition of the potential function, 
\begin{align*}
\Phi(s^{r}) - \Phi(s^{r + 1}) \; =& \; C_{i^{r}}(s^{r}) -
C_{i^{r}}(\widehat{s}_{i^{r}}^{\, r}, s_{-i^{r}}^{r}) \\
	\geq& \; \frac{1}{\underline{\eta}}\widehat{C}_{i^{r}}(s^{r}) - \overline{\eta}\widehat{C}_{i^{r}}(\widehat{s}_{i^{r}}^{\, r}, s_{-i^{r}}^{r}) \\
	>& \; \frac{1}{\underline{\eta}} (\underline{\eta}\overline{\eta}) \widehat{C}_{i^{r}}(\widehat{s}_{i^{r}}^{\, r}, s_{-i^{r}}^{r}) - \overline{\eta}\widehat{C}_{i^{r}}(\widehat{s}_{i^{\, r}}^{r}, s_{-i^{r}}^{r}) \\
	=& \; 0 \, .
\end{align*}
The second formula follows from the definition of the $\epsilon$-individual cost. The third one follows from Eq.~(\ref{formula_condition_of_being_selected_for_updating_the_strategy}).
\end{proof}


\begin{theorem}
\label{theorem:approx-ratio-Alg-ABRD}
Let
$Q = \frac{2 (\underline{\eta}\overline{\eta}) N}{1 - \varrho (\underline{\eta}\overline{\eta})^{2} \mu}$.
If
$R =
\left\lceil
Q \cdot \ln \left( K N^{\alpha_{\max} - 1} \right)
\right\rceil$,
then the output $s^{\texttt{out}}$ of \texttt{Bayes-ABRD} satisfies
\begin{equation*}
C(s^{\texttt{out}})
\leq
\frac{2 K \varrho (\underline{\eta}\overline{\eta})^{2} \lambda}{1 - \varrho (\underline{\eta}\overline{\eta})^{2} \mu}
\cdot \mathbb{E}_{t \sim T}\Big[\mathtt{OPT}(t)\Big] \, .
\end{equation*}
\end{theorem}


\begin{proof}
Lemma \ref{lemma:ABRD-converge-approx-NE} ensures that
the assertion holds if the ABRD simulated in \texttt{Bayes-ABRD} converges in any round 
$r \leq R$,
so it is left to analyze the case where the ABRD does not converge.
We say a profile $s^{r}$ involved in the ABRD is \emph{bad} if
\[
C(s^{r})
\, > \,
\frac{2 \varrho (\underline{\eta}\overline{\eta})^{2} \lambda}{1 - \varrho (\underline{\eta}\overline{\eta})^{2} \mu}
\cdot \mathbb{E}_{t \sim T}\Big[\mathtt{OPT}(t)\Big] \, .
\]

\begin{claim}
\label{proposition_potential_function_decreases_significantly_in_the_steps_with_bad_states}
For any
$r < R$,
if $s^{r}$ is bad, then
$\Phi(s^{r + 1}) < (1 - 1 / Q) \cdot \Phi(s^{r})$.
\end{claim}
\begin{subproof}
\qedlabel{proposition_potential_function_decreases_significantly_in_the_steps_with_bad_states}
Fix
\begin{equation}
\label{formula_difference_definition_of_nabla_tau}
d^{r}
\; = \;
\overline{\eta} \Big[
\sum_{i \in [N]} \widehat{C}_{i}(s^{r}) -
(\underline{\eta}\overline{\eta})  \sum_{i \in [N]} \widehat{C}_{i}(\widehat{s}_{i}^{\, r}, s_{-i}^{r})
\Big] \, .
\end{equation}
This means that
\begin{align*}
C(s^{r})
\; = \;
\sum_{i \in [N]} C_{i}(s^{r})
\; \leq & \;
\overline{\eta} \sum_{i \in [N]} \widehat{C}_{i}(s^{r}) \\
= & \;
\overline{\eta}(\underline{\eta}\overline{\eta})  \sum_{i \in [N]} \widehat{C}_{i}(\widehat{s}_{i}^{\, r}, s_{-i}^{r}) + d^{r} \\
\leq & \;
\varrho (\underline{\eta}\overline{\eta})^{2} \sum_{i \in [N]} C_{i}(s_{i}', s_{-i}^{r}) + d^{r} \\
\leq & \;
\varrho (\underline{\eta}\overline{\eta})^{2} \Big(\lambda \cdot \mathbb{E}_{t \sim T}\Big[\mathtt{OPT}(t)\Big] + \mu C(s^{r}) \Big) +
d^{r} \, .
\end{align*}
Therefore,
$d^{r}
\geq
\left[ 1 - \varrho (\underline{\eta}\overline{\eta})^{2} \mu \right] C(s^{r}) -
\varrho (\underline{\eta}\overline{\eta})^{2} \lambda \cdot \mathbb{E}_{t \sim T}\Big[\mathtt{OPT}(t)\Big]$,
hence, if $s^{r}$ is bad, then $d^{r}$ satisfies
\begin{equation}
\label{formula_lower_bound_on_nabla_tau}
d^{r}
\, > \,
\left[ 1 - \varrho (\underline{\eta}\overline{\eta})^{2} \mu \right] C(s^{r}) -
\frac{1 - \varrho (\underline{\eta}\overline{\eta})^{2} \mu}{2} C(s^{r})
\, = \,
\frac{1 - \varrho (\underline{\eta}\overline{\eta})^{2} \mu}{2} C(s^{r}) \, .
\end{equation}

Since the ABRD does not converge at round $r$, there exists a player $i^{r}$
being selected to update her strategy.
Recalling that the ABR of player $i$ to $s^{r}$ is
denoted by $\widehat{s}_{i}^{\, r}$, we observe that
\begin{align*}
\Phi(s^{r}) - \Phi(s^{r + 1})
\; = & \;
C_{i^{r}}(s^{r}) - C_{i^{r}}(\widehat{s}_{i^{r}}^{\, r}, s_{-i^{r}}^{r}) \\
\geq & \;
\frac{1}{\underline{\eta}} \widehat{C}_{i^{r}}(s^{r}) - \overline{\eta}
\cdot \widehat{C}_{i^{r}}(\widehat{s}_{i^{r}}^{\, r}, s_{-i^{r}}^{r}) \\
= & \;
\frac{1}{\underline{\eta}} \left[ \widehat{C}_{i^{r}}(s^{r}) - (\underline{\eta}\overline{\eta})
\widehat{C}_{i^{r}}(\widehat{s}_{i^{r}}^{\, r}, s_{-i^{r}}^{r}) \right] \\
\geq & \;
\frac{1}{\underline{\eta}} \cdot \frac{1}{N} \sum_{i \in [N]} \left[
\widehat{C}_{i}(s^{r}) - (\underline{\eta}\overline{\eta})
\widehat{C}_{i}(\widehat{s}_{i}^{\, r}, s_{-i}^{r}) \right] \\
= & \;
\frac{1}{\underline{\eta}\overline{\eta}} \cdot \frac{d^{r}}{N} \\
> & \;
\frac{1}{\underline{\eta}\overline{\eta}} \cdot \frac{1}{2N} \left[ 1 - \varrho (\underline{\eta}\overline{\eta})^{2} \mu \right] C(s^{r}) \\
\geq & \;
\frac{1}{\underline{\eta}\overline{\eta}} \cdot \frac{1}{2N} \left[ 1 - \varrho (\underline{\eta}\overline{\eta})^{2} \mu
\right] \Phi(s^{r}) \, ,
\end{align*}
where
the fourth transition follows from
Eq.~\eqref{formula_condition_of_being_selected_for_updating_the_strategy},
the fifth and sixth transitions follow from
Eq.~(\ref{formula_difference_definition_of_nabla_tau}) and
Eq.~(\ref{formula_lower_bound_on_nabla_tau}), respectively,
and
the last transition holds because the potential function is assumed to be $K$-bounded.
Therefore,
\[
\Phi(s^{r + 1})
\, < \,
\Phi(s^{r}) \left( 1 - \frac{1 - \varrho (\underline{\eta}\overline{\eta})^{2} \mu}{2 (\underline{\eta}\overline{\eta}) N} \right)
\, = \,
(1 - 1 / Q) \cdot \Phi(s^{r})
\]
as promised.
\end{subproof}

\begin{claim}
\label{proposition_decrease_in_potential_function_when_all_states_are_bad}
Assuming that all the
$R + 1$
strategy profiles in the ABRD are bad, we have
$C(s^{R}) < \varrho \cdot \mathbb{E}_{t \sim T}\Big[\mathtt{OPT}(t)\Big]$.
\end{claim}
\begin{subproof}
\qedlabel{proposition_decrease_in_potential_function_when_all_states_are_bad}
Claim~\ref{proposition_potential_function_decreases_significantly_in_the_steps_with_bad_states}
implies that if all the $R + 1$ profiles involved in the ABRD are bad, then
\[
\Phi(s^{R})
\, < \,
\left( 1 - \frac{1}{Q} \right)^{R} \Phi(s^{0})
\, = \,
\left( 1 - \frac{1}{Q} \right)^{\left\lceil Q \cdot \ln \Big( K
N^{\alpha_{\max} - 1} \Big) \right\rceil} \Phi(s^{0})
\, \leq \,
\frac{1}{K N^{\alpha_{\max} - 1}} \Phi(s^{0}) \, .
\]
By the definition of the bounded potential function and by
Lemma \ref{lemma:optimal-reply-guarantees-approx-ratio},
we have
\begin{equation*}
C(s^{R})
\, \leq \, 
K \cdot \Phi(s^{R})
\, < \, 
\frac{K \Phi(s^{0})}{K N^{\alpha_{\max} - 1}} 
\, \leq \, 
\frac{C(s^{0})}{N^{\alpha_{\max} - 1}} 
\, \leq\, 
\frac{\varrho N^{\alpha_{\max} - 1} \mathbb{E}_{t \sim T}\Big[\mathtt{OPT}(t)\Big]}{N^{\alpha_{\max} - 1}} \, ,
\end{equation*}
which completes the proof.
\end{subproof}

\begin{claim}
\label{proposition_general_lower_bound_on_smoothness_parameters}
$\varrho
<
\frac{2\varrho (\underline{\eta}\overline{\eta})^{2} \lambda}
{1 - \varrho (\underline{\eta}\overline{\eta})^{2}\mu}$.
\end{claim}
\begin{subproof}
\qedlabel{proposition_general_lower_bound_on_smoothness_parameters}
It can be inferred from \cite{Roughgarden2012POAIC} that the parameters $(\lambda, \mu)$ should satisfy $\frac{\lambda}{1 - \mu} \geq 1$ if the game is $(\lambda, \mu)$-smooth.
Therefore,
$\frac{2 \varrho (\underline{\eta}\overline{\eta})^{2} \lambda}{1 - \varrho (\underline{\eta}\overline{\eta})^{2} \mu}
>
\frac{2\varrho \lambda}{1 - \mu}
>
\varrho$.
\end{subproof}

By Claim \ref{proposition_general_lower_bound_on_smoothness_parameters} and the definition of bad strategy profiles, $C(s^{R}) < \varrho \cdot \mathbb{E}_{t \sim T}\Big[\mathtt{OPT}(t)\Big]$ implies that $s^{R}$ is not bad, which conflicts with the assumption of Claim \ref{proposition_decrease_in_potential_function_when_all_states_are_bad} that all the $R + 1$ strategy profiles are bad. This means that there exists at least one round $r^{*}$ whose corresponding strategy profile $s^{r^{*}}$ is not bad. 
Therefore,
\begin{align*}
C(s^{R}) \; \leq \; K \cdot \Phi(s^{R}) \; \leq \; K \cdot \Phi(s^{r^{*}}) \; \leq \; K \cdot C(s^{r^{*}}) \; \leq \; K \cdot \frac{2 \varrho (\underline{\eta}\overline{\eta})^{2} \lambda}{1 - \varrho (\underline{\eta}\overline{\eta})^{2} \mu}
\cdot \mathbb{E}_{t \sim T}\Big[\mathtt{OPT}(t)\Big] \, .
\end{align*}
The first transition and the third one holds because the potential function is $K$-bounded. The second transition follows from Lemma \ref{lemma_potential_function_decreases_in_each_step_of_ABRD}. The last transition holds because $s^{r^{*}}$ is not bad. This completes the proof.
\end{proof}

\section{Smoothness Parameters}
\label{section:smoothness}
In this section, we consider the case where the parameters $\varrho$, $\underline{\eta}$ and $\overline{\eta}$ are fixed, and focus on finding proper parameters $(\lambda, \mu)$ such that the BGND game is $(\lambda, \mu)$-smooth, and $\mu < 1 / [\varrho (\underline{\eta} \overline{\eta})^{2}]$.


\begin{lemma}\label{lemma_condition_for_smoothness_parameters}
For any pair of parameters $\lambda' > 0$ and $0 < \mu' < 1 \Big/ \Big[\varrho \big(\underline{\eta} \overline{\eta} \big)^{2}\Big]$, if
\begin{equation}
y \cdot (x + y)^{\alpha_{j} - 1} \; \leq \; \lambda' \cdot y^{\alpha_{j}} + \mu' \cdot x^{\alpha_{j}} \label{formula_condition_for_smoothness_parameters}
\end{equation}
holds for any $x, y \in \mathbb{R}_{\geq 0}$ and every $j \in [q]$, then the BGND game is $(\lambda', \mu')$-smooth.
\end{lemma}

\begin{proof}
This proposition can be proved in a similar way as \cite{Roughgarden2012POAIC}. For any resource $e$ and any type $t$, we say $e \in \sigma^{*}(t)$ if there exists some player $i$ such that $e \in \sigma_{i}^{*}(t)$. 
For any type profiles $t, t'$, every action profile $a \in A^{t'}$, and every resource $e \in \sigma^{*}(t)$,
\begin{align*}
\sum_{i \in [N]} f_{i, e}(t_{i}; (\sigma_{i}^{*}(t), a_{-i})) \; =& \; \sum_{i \in [N]: e \in \sigma_{i}^{*}(t)}\sum_{j \in [q]}\xi_{e, j}\Big(l_{e}^{(\sigma_{i}^{*}(t), a_{-i})}\Big)^{\alpha_{j} - 1} \\
\leq& \; |l_{e}^{\sigma^{*}(t)}|\sum_{j \in [q]}\xi_{e, j}\Big(l_{e}^{\sigma^{*}(t)} + l_{e}^{a}\Big)^{\alpha_{j} - 1} \\
\leq& \; \sum_{j \in [q]}\xi_{e, j}\Big[\lambda' \cdot \Big(l_{e}^{\sigma^{*}(t)} \Big)^{\alpha_{j}} + \mu' \cdot \Big(l_{e}^{a} \Big)^{\alpha_{j}} \Big] \, ,
\end{align*}
where the third transition follows from Eq.~\eqref{formula_condition_for_smoothness_parameters}. Then
\begin{align*}
\sum_{i \in [N]} C_{i}(t_{i}; (\sigma_{i}^{*}(t), a_{-i})) \; =& \; \sum_{i \in [N]}\sum_{e \in E} f_{i, e}(t_{i}; (\sigma_{i}^{*}(t), a_{-i})) \\
=& \; \sum_{e \in \sigma^{*}(t)}\sum_{i \in [N]} f_{i, e}(t_{i}; (\sigma_{i}^{*}(t), a_{-i})) \\
\leq& \; \lambda' \cdot \sum_{e \in \sigma^{*}(t)}\sum_{j \in [q]}\xi_{e, j}\Big(l_{e}^{\sigma^{*}(t)} \Big)^{\alpha_{j}} + \mu' \cdot \sum_{e \in \sigma^{*}(t)} \sum_{j \in [q]}\xi_{e, j}\Big(l_{e}^{a} \Big)^{\alpha_{j}} \\
\leq& \; \lambda' \cdot \mathtt{OPT}(t) + \mu' \cdot \sum_{i \in [N]}C_{i}(t_{i}', a) \, .
\end{align*}
The second transition above holds because for any $e \notin \sigma^{*}(t)$, $f_{i, e}(t_{i}; (\sigma_{i}^{*}(t), a_{-i})) = 0$ for every player $i$. The last transition holds because $l_{e}^{a} = 0$ for any $e \notin a$, which implies that $\sum_{e \in \sigma^{*}(t)} \sum_{j \in [q]}\xi_{e, j}\Big(l_{e}^{a} \Big)^{\alpha_{j}} = \sum_{e \in \sigma^{*}(t) \cap a} \sum_{j \in [q]}\xi_{e, j}\Big(l_{e}^{a} \Big)^{\alpha_{j}}$.
\end{proof}

\begin{lemma}[\cite{Aland2006EPA}]\label{lemma_condition_on_smoothness_parameter_lambda}
For any $\mu' \in (0, \frac{1}{\varrho(\underline{\eta}\overline{\eta})^{2}} )$, setting $\lambda' = \max_{x \in \mathbb{R}_{> 0}}(x + 1)^{\alpha_{\max} - 1} - \mu' \cdot x^{\alpha_{\max}}$ satisfies Eq.~\eqref{formula_condition_for_smoothness_parameters}.
\end{lemma}

For any $\mu' \in (0, \frac{1}{\varrho(\underline{\eta}\overline{\eta})^{2}} )$, define $g_{\mu'}(x) = (x + 1)^{\alpha_{\max} - 1} - \mu' \cdot x^{\alpha_{\max}}$ and $h(x) = \Big[ (\alpha_{\max} - 1)(x + 1)^{\alpha_{\max} - 2} \Big]\Big/\Big[ \alpha_{\max} \cdot x^{\alpha_{\max} - 1} \Big]$.  Then:

\begin{lemma}[\cite{Aland2006EPA}]\label{lemma_global_maximum_point_for_smoothness_parameter_lambda}
For any $\mu' \in (0, \frac{1}{\varrho(\underline{\eta}\overline{\eta})^{2}} )$, there exists a unique real positive number $x_{\mu'}$ that maximizes $g_{\mu'}(x)$ over $x \in (0, +\infty)$, and this number satisfies $h_{\alpha_{\max}}(x_{\mu'}) = \mu'$.
\end{lemma}

\begin{remark}
Note that in \cite{Aland2006EPA}, the propositions corresponding to Lemma \ref{lemma_condition_on_smoothness_parameter_lambda} and Lemma \ref{lemma_global_maximum_point_for_smoothness_parameter_lambda} are proved for the case $\alpha_{\max} \in \mathbb{Z}_{\geq 2}$, but their proofs directly hold for the case where $\alpha_{\max} \in \mathbb{R}_{\geq 1}$.
\end{remark}

It can be inferred from the derivative that $(x + 1)^{\alpha_{\max} - 1} = x^{\alpha_{\max}}$ has a unique positive root, which is denoted by $\gamma_{\alpha_{\max}}$. In \cite{Aland2006EPA}, it is proved that if $\alpha_{\max} \geq 3$, $\gamma_{\alpha_{\max}}$ is bounded by $O\Big(\frac{\alpha_{\max} - 1}{\ln(\alpha_{\max} - 1)}\Big)$. Furthermore, it can be verified that $\gamma_{\alpha_{\max}} < 2$ when $\alpha_{\max} \in (1, 2)$, and $\gamma_{\alpha_{\max}} < 3$ when $\alpha_{\max} \in [2, 3)$.

Let $\mu_{\alpha} = h\Big( \varrho(\underline{\eta}\overline{\eta})^{2} \cdot \gamma_{\alpha_{\max}} \Big)$, and $\lambda_{\alpha} = g_{\mu_{\alpha}}\Big( \varrho(\underline{\eta}\overline{\eta})^{2} \cdot \gamma_{\alpha_{\max}} \Big)$. Then we have:

\begin{theorem}\label{theorem_desired_smoothness_parameters_for_BGND_game}
The BGND game is $(\lambda_{\alpha}, \mu_{\alpha})$-smooth, and $\varrho(\underline{\eta}\overline{\eta})^{2}\mu_{\alpha} < 1 - 1 / \alpha_{\max}$.
\end{theorem}


\begin{proof} 
\begin{align*}
h\Big( \varrho(\underline{\eta}\overline{\eta})^{2} \cdot \gamma_{\alpha_{\max}} \Big) \; =& \; \frac{(\alpha_{\max} - 1) \cdot \Big( \varrho(\underline{\eta}\overline{\eta})^{2} \cdot \gamma_{\alpha_{\max}} + 1\Big)^{\alpha_{\max} - 2}}{\alpha_{\max} \cdot \Big( \varrho(\underline{\eta}\overline{\eta})^{2} \cdot \gamma_{\alpha_{\max}}\Big)^{\alpha_{\max} - 1}} \\
=& \; \frac{(\alpha_{\max} - 1) \cdot \Big( \varrho(\underline{\eta}\overline{\eta})^{2} \cdot \gamma_{\alpha_{\max}} + 1\Big)^{\alpha_{\max} - 1}  \Big( \varrho(\underline{\eta}\overline{\eta})^{2} \cdot \gamma_{\alpha_{\max}}\Big)}{\alpha_{\max} \cdot \Big( \varrho(\underline{\eta}\overline{\eta})^{2} \cdot \gamma_{\alpha_{\max}}\Big)^{\alpha_{\max}} \Big( \varrho(\underline{\eta}\overline{\eta})^{2} \cdot \gamma_{\alpha_{\max}} + 1\Big)} \\
=& \; \frac{(\alpha_{\max} - 1)  \Big( \varrho(\underline{\eta}\overline{\eta})^{2} \cdot \gamma_{\alpha_{\max}} + 1\Big)^{\alpha_{\max} - 1}  \Big( \varrho(\underline{\eta}\overline{\eta})^{2} \cdot \gamma_{\alpha_{\max}}\Big)}{\alpha_{\max}  \varrho(\underline{\eta}\overline{\eta})^{2} \Big[  \varrho(\underline{\eta}\overline{\eta})^{2} \cdot \Big( \gamma_{\alpha_{\max}} + 1\Big)\Big]^{\alpha_{\max} - 1} \Big( \varrho(\underline{\eta}\overline{\eta})^{2} \cdot \gamma_{\alpha_{\max}} + 1\Big)}\\
<& \; \frac{(\alpha_{\max} - 1)}{\alpha_{\max} \varrho(\underline{\eta}\overline{\eta})^{2}} \, ,
\end{align*}
which implies that $\varrho(\underline{\eta}\overline{\eta})^{2} \cdot \mu_{\alpha} < (\alpha_{\max} - 1)/\alpha_{\max}$. The third transition holds because $\gamma_{\alpha_{\max}}$ is the positive root of $(x + 1)^{\alpha_{\max} - 1} - x^{\alpha_{\max}} = 0$. The fourth transition holds because $\varrho(\underline{\eta}\overline{\eta})^{2} \geq 1$, $\gamma_{\alpha_{\max}} > 0$ and $\alpha_{\max} - 1 > 0$. 

Since $\mu_{\alpha} < \frac{\alpha_{\max} - 1}{\alpha_{\max}} \cdot \frac{1}{\varrho(\underline{\eta}\overline{\eta})^{2}} < \frac{1}{\varrho(\underline{\eta}\overline{\eta})^{2}}$, it can be inferred from Lemma \ref{lemma_global_maximum_point_for_smoothness_parameter_lambda} that $g_{ \mu_{\alpha}}\Big( \varrho \cdot (\underline{\eta}\overline{\eta})^{2} \cdot \gamma_{\alpha_{\max}} \Big) = \max_{x > 0}(x + 1)^{\alpha_{\max} - 1} - \mu_{\alpha} \cdot x^{\alpha_{\max}}$. It implies that $\lambda_{\alpha}$ satisfies the condition given in Lemma \ref{lemma_condition_on_smoothness_parameter_lambda}. By Lemma \ref{lemma_condition_for_smoothness_parameters}, this theorem holds.
\end{proof}


\section{Bounded Potential Function}
\label{section:bounded-potential}
In this section, it is proved that the BGND game admits a potential function that is $K$-bounded with $K = \lceil \alpha_{\max} \rceil$.

\begin{lemma}[\cite{Alon2012BI}]
If for every type profile $t$, there exists a function $\Phi_{t}: A \mapsto \mathbb{R}_{\geq 0}$ such that for every action profile $a \in A^{t}$, every $i \in [N]$ and every $a_{i}' \in A_{i}^{t_{i}}$, 
\begin{equation}
\Phi_{t}(a) - \Phi_{t}(a_{i}', a_{-i}) = C_{i}(t_{i}; a) - C_{i}(t_{i}; (a_{i}', a_{-i}) ) \, , \label{formula_type_specified_potential_function}
\end{equation}
then $\Phi(s) = \sum_{t \in T}p(t)\Phi_{t}(s(t))$ is a potential function of the BGND game. 
\end{lemma}


\begin{theorem}\label{theorem_existence_of_bounded_potential}
For the BGND game, there exists a potential function $\Phi(s)$ that is $\lceil \alpha_{\max}\rceil$-bounded.
\end{theorem}

\begin{proof}
For every type profile $t$ and every action profile $a \in A^{t}$, define the function
\begin{equation*}
\Phi_{t}(a) \; = \; \sum_{e \in E}\sum_{l = 1}^{l_{e}^{a}}\sum_{j \in [q]}\xi_{e, j} \cdot l^{\alpha_{j} - 1} \, .
\end{equation*}
In \cite{Monderer1996PG}, it is proved that such a function satisfies Eq.~\eqref{formula_type_specified_potential_function}, which implies that the BGND game admits a potential function $\Phi(s) = \sum_{t \in T}p(t)\Phi_{t}(s(t))$. Furthermore, for every $e \in E$ and every $j \in [q]$, $\sum_{l = 1}^{l_{e}^{a}}l^{\alpha_{j} - 1} \leq (l_{e}^{a})^{\alpha_{j}}$ trivially holds, and
\begin{equation*}
\sum_{l = 1}^{l_{e}^{a}}l^{\alpha_{j} - 1} \; \geq \; \frac{1}{(l_{e}^{a})^{\lceil \alpha_{j} \rceil - \alpha_{j}}}\sum_{l = 1}^{l_{e}^{a}}l^{\lceil \alpha_{j} \rceil - 1} \; \geq \; \frac{1}{(l_{e}^{a})^{\lceil \alpha_{j} \rceil - \alpha_{j}}} \cdot \frac{1}{\lceil \alpha_{j} \rceil} (l_{e}^{a})^{\lceil \alpha_{j} \rceil} \; = \; \frac{1}{\lceil \alpha_{j} \rceil} (l_{e}^{a})^{\alpha_{j}} \, ,
\end{equation*}
where the second transition follows from \cite{Anshelevich2008PSN}. Therefore, it can be obtained that $\Phi_{t}(a) \leq \sum_{e \in E}F_{e}(l_{e}^{a}) \leq \lceil \alpha_{\max} \rceil \Phi_{t}(a)$. By the linearity of expectation, $\Phi(s)$ is $\lceil \alpha_{\max} \rceil$-bounded.
\end{proof}

\section{Efficient Estimation of the Cost Share}
\label{section:efficient-estimation-cost-share}
This section focuses on the $(\underline{\eta}, \overline{\eta})$-estimation of the expected cost shares. 
For any
$z \in (0, 1)$
and
$z' \geq 1$,
define
$b_{z} = \Big((\beta^{\circ})^{2} + 1\Big)\Big(1 - \frac{1}{\beta^{\circ}}
\Big)^{-z}$
with $\beta^{\circ}$ being the unique root of
$2\beta^{3} - (z + 2)\beta^{2} - 2 = 0$
in the interval
$(1, +\infty)$,
$B_{z'}$ to be the fractional Bell number with the parameter $z'$
\cite{Bampis2014EES, Makarychev2014SOP},
and $\gamma_{z'}$ to be the unique positive root of
$(x + 1)^{z' - 1} = x^{z'}$
\cite{Aland2006EPA}. For any $i \in [N]$, $e \in E$ and any $s_{-i}$, it is shown that there exists a $(\max\lbrace 1, \max_{\alpha_{j} \in (1, 2)}b_{\alpha_{j} - 1}\rbrace,\, \max\lbrace \max_{\alpha_{j} \geq 2} B_{\alpha_{j} - 1}, 1 \rbrace)$-estimation $\widehat{\mathfrak{f}}_{i, e}(+; s_{-i})$ of $\mathfrak{f}_{i, e}(+; s_{-i})$ that can be obtained in $\operatorname{poly}(q, N, \lbrace |T_{i}| \rbrace_{i \in [N]})$-time. Particularly, consider the special case of the BGND game where $q = 1$ and $\alpha_{1} = 2$, which means that the cost function associated with each resource $e \in E$ can be written as
\begin{equation}
F_{e}(l) = \xi_{e} \cdot l^{2} \, . \label{formula_quadratic_edge_cost_function}
\end{equation}
It is proved that the BGND game with such a quadratic cost function is tractable.

\begin{lemma}[\cite{Berend2010IBB}]\label{lemma_upper_bound_on_the_expectation_of_power_of_sum_of_random_variables}
Let $\lbrace X_{1}, X_{2}, \cdots, X_{k}, \cdots \rbrace$ be a finite set of mutually independent random variables following the Bernoulli distribution supported on $\lbrace 0, 1 \rbrace$. Then for any $z \geq 1$,
\begin{equation}
\mathbb{E}\Big[ \big(\sum_{k} X_{k} \big)^{z} \Big] \; \leq \; B_{z} \cdot \max\Big\lbrace \mathbb{E}\big[ \sum_{k} X_{k}\big],\, \Big( \mathbb{E}\big[ \sum_{k} X_{k} \big]\Big)^{z} \Big\rbrace \, .\nonumber
\end{equation}
\end{lemma}

\begin{lemma}\label{lemma_bounds_on_expectation_of_concave_monomial}
Let $\lbrace X_{1}, X_{2}, \cdots, X_{k}, \cdots \rbrace$ be a finite set of Bernoulli random variables that are mutually independent. For any $z' \in (0, 1)$ and $\beta > 1$:
\begin{equation*}
\frac{1}{\beta^{2} + 1}\Big(1 - \frac{1}{\beta} \Big)^{z'} \Big( \mathbb{E}\Big[1 + \sum_{k}X_{k}\Big] \Big)^{z'} \; \leq \; \mathbb{E}\Big[\Big(1 + \sum_{k}X_{k}\Big)^{z'}\Big] \; \leq \; \Big(\mathbb{E}\Big[1 + \sum_{k}X_{k}\Big]\Big)^{z'} \, .
\end{equation*}
\end{lemma}

\begin{proof}
The expression $\mathbb{E}\Big[\Big(1 + \sum_{k}X_{k}\Big)^{z'}\Big] \leq \Big( \mathbb{E}\Big[1 + \sum_{k}X_{k}\Big] \Big)^{z'}$ follows from Jensen's inequality \cite{Jensen1906SLF}, since the function $\varphi(x) = x^{z'}$ is concave when $z' \in (0, 1)$. Now consider the lower bound on $\mathbb{E}\Big[\Big(1 + \sum_{k}X_{k}\Big)^{z'}\Big]$. Let $\operatorname{Var}\Big[ 1 + \sum_{k}X_{k} \Big]$ be the variance of the random variable $1 + \sum_{k}X_{k}$. Then we have
\begin{align*}
\operatorname{Var}\Big[ 1 + \sum_{k}X_{k} \Big] \; =& \; \mathbb{E}\bigg[ \Big(1 + \sum_{k}X_{k}\Big)^2 \bigg] - \Big(\mathbb{E}\Big[ 1 + \sum_{k}X_{k} \Big]\Big)^{2} \\
\leq& \; B_{2} \cdot \max\bigg\lbrace \mathbb{E}\Big[ 1 + \sum_{k}X_{k} \Big],\, \Big(\mathbb{E}\Big[ 1 + \sum_{k}X_{k} \Big]\Big)^{2} \bigg\rbrace - \Big(\mathbb{E}\Big[ 1 + \sum_{k}X_{k} \Big]\Big)^{2} \\
=& \; 2 \cdot \Big(\mathbb{E}\Big[ 1 + \sum_{k}X_{k} \Big]\Big)^{2} - \Big(\mathbb{E}\Big[ 1 + \sum_{k}X_{k} \Big]\Big)^{2} \\
=& \; \Big(\mathbb{E}\Big[ 1 + \sum_{k}X_{k} \Big]\Big)^{2} \, .
\end{align*}
The second transition above follows from Lemma \ref{lemma_upper_bound_on_the_expectation_of_power_of_sum_of_random_variables}, because $\lbrace X_{k} \rbrace$ are mutually independent Bernoulli random variables, and the constant $1$ can also be viewed as a Bernoulli random variable which equals to $1$ with probability $1$. The third transition holds because $1 + \sum_{k}X_{k} \geq 1$. By using Cantelli's  inequality \cite{Savage1961PIT}, it can be obtained that for any $\beta > 1$,
\begin{equation*}
\operatorname{Pr}\bigg[ \Big( 1 + \sum_{k}X_{k} \Big) < \Big(1 - \frac{1}{\beta}\Big)\mathbb{E}\Big[ 1 + \sum_{k}X_{k} \Big]\bigg] \, \leq \, \frac{1}{1 + \frac{\Big(\mathbb{E}\Big[ 1 + \sum_{k}X_{k} \Big] \Big)^{2}}{\beta^{2}\cdot \operatorname{Var}\Big[ 1 + \sum_{k}X_{k} \Big]} } \,
= \, \frac{\beta^2}{\beta^2 + 1} \, ,
\end{equation*}
where $\operatorname{Pr}[\cdot]$ denotes the probability of random events. Therefore,
\begin{align*}
& \; \mathbb{E}\Big[\Big(1 + \sum_{k}X_{k}\Big)^{z'}\Big] \\
\geq& \; \operatorname{Pr}\bigg[ \Big( 1 + \sum_{k}X_{k} \Big) \geq \Big(1 - \frac{1}{\beta}\Big)\mathbb{E}\Big[ 1 + \sum_{k}X_{k} \Big]\bigg] \cdot \bigg[ \Big(1 - \frac{1}{\beta}\Big)\mathbb{E}\Big[ 1 + \sum_{k}X_{k} \Big] \bigg]^{z'} \\
\geq& \; \Big( 1 - \frac{\beta^2}{\beta^2 + 1} \Big)\Big(1 - \frac{1}{\beta}\Big)^{z'}\Big(\mathbb{E}\Big[1 + \sum_{k}X_{k}\Big]\Big)^{z'} \, .
\end{align*}
\end{proof}

Recalling that $b_{z'} = \Big((\beta^{\circ})^{2} + 1\Big)\Big(1 - \frac{1}{\beta^{\circ}} \Big)^{-z'}$ with $\beta^{\circ}$ being the unique root of $2\beta^{3} - (z + 2)\beta^{2} - 2 = 0$ in the interval $(1, +\infty)$, we have the following lemma.

\begin{lemma}\label{lemma_uniqueness_of_the_lower_bound_parameter_b_alpha}
For any $z' \in (0, 1)$, $b_{z'} = \min_{\beta > 1} (\beta^{2} + 1)\Big(1 - \frac{1}{\beta} \Big)^{-z'}$.
\end{lemma}


\begin{proof}
Let $\varphi(\beta) = (\beta^{2} + 1)\Big(1 - \frac{1}{\beta} \Big)^{-z'}$. Fix $z'$, the derivative of $\varphi$ with respect to $\beta$ is
\begin{equation*}
\frac{\mathrm{d \varphi}}{\mathrm{d} \beta} \; = \; \Big( 1 - \frac{1}{\beta} \Big)^{-z' - 1}\frac{1}{\beta^{2}}\Big(2\beta^{3} - (2 + z')\beta^{2} - 2 \Big) \, .
\end{equation*}
It can be further derived from the derivative that $2\beta^{3} - (2 + z')\beta^{2} - 2$ is monotonically increasing in the interval $(\frac{2 + z'}{3}, \infty)$. Since $\frac{2 + z'}{3} < 1$, $2\beta^{3} - (2 + z')\beta^{2} - 2 < 0$ for $\beta = 1$, and $2\beta^{3} - (2 + z')\beta^{2} - 2 > 0$ for $\beta = 2$, there exists a unique $\beta^{\circ} \in (1, +\infty)$ such that $2\beta^{3} - (2 + z')\beta^{2} - 2 = 0$, and $\beta^{\circ}$ minimizes $(\beta^{2} + 1)\Big(1 - \frac{1}{\beta} \Big)^{-z'}$ because $\Big( 1 - \frac{1}{\beta} \Big)^{-z' - 1}\frac{1}{\beta^{2}} > 0$ for any $\beta > 1$.
\end{proof}


For each action $a_{i}$ of player $i$ and each resource $e$, denote the indicator of whether $e$ is contained in $a_{i}$ by $\delta(a_{i}, e)$. Formally,
\begin{equation*}
\delta(a_{i}, e) \; = \; \begin{cases}
0 & \text{if } e \notin a_{i} \\
1 & \text{otherwise}
\end{cases}
\end{equation*}

\begin{theorem}\label{theorem_bounds_on_estimation_for_expected_individual_cost}
For any player $i$, any edge $e$, any action $a_{i}$, and any strategies $s_{-i}$, let
\begin{equation}
\widehat{\mathfrak{f}}_{i, e}(+; s_{-i}) \; = \;
\sum_{j \in [q]}\xi_{e, j}\Big[ 1 + \sum_{i' \neq i}\sum_{t_{i'} \in T_{i'}}p_{i'}(t_{i'})\delta\big(s_{i'}(t_{i'}), e\big) \Big]^{\alpha_{j} - 1} 
 \, , \label{formula_approximate_individual_cost}
\end{equation}
then 
\begin{equation}
\widehat{\mathfrak{f}}_{i, e}(+; s_{-i}) / \max\Big\lbrace 1,\, \max_{j: \alpha_{j} \in (1, 2)} b_{\alpha_{j} - 1} \Big\rbrace \, \leq \, \mathfrak{f}_{i, e}(+; s_{-i}) \; \leq \; \widehat{\mathfrak{f}}_{i, e}(+; s_{-i}) \cdot \Big\lbrace 1,\, \max_{j : \alpha_{j} \geq 2} B_{\alpha_{j} - 1} \Big\rbrace \, .\label{formula_bounds_on_estimation_parameters_for_general_alpha_j}
\end{equation}
In particular, if for every resource $e$, $F_{e}(l)$ is a quadratic function given in Eq.~\eqref{formula_quadratic_edge_cost_function}, then $\widehat{\mathfrak{f}}_{i, e}(+; s_{-i}) = \mathfrak{f}_{i, e}(+; s_{-i})$. 
\end{theorem}

\begin{proof}
Let $a_{i}$ be an action in $A_{i}$ satisfying $e \in a_{i}$. By definition, we have
\begin{align*}
\mathfrak{f}_{i, e}(+; s_{-i})
\; = & \;
\mathbb{E}_{t_{-i} \sim p_{-i}} \Left[ f_{i, e}(a_{i}, s_{-i}(t_{-i})) \Right] \\
=& \;
\sum_{j \in [N]}\xi_{e, j} \cdot \mathbb{E}_{t_{-i} \sim p_{-i}} \Left[
\Left( l_{e}^{a_{i}, s_{-i}(t_{-i})} \Right)^{\alpha_{j} - 1} \Right] \\
=& \;
\sum_{j \in [N]}\xi_{e, j} \cdot \mathbb{E}_{t_{-i} \sim p_{-i}} \Left[ \Left(
1 + \sum_{i' \in [N]: i' \neq i}\delta(s_{i'}(t_{-i}(i')), e) \Right)^{\alpha_{j} -
1} \Right] \\
=& \;
\sum_{j \in [N]} \xi_{e, j} \mathbb{E}_{ \lbrace t_{i'} \sim p_{i'} \rbrace_{
i' \neq i} } \Left[ \Left( 1 + \sum_{i' \neq i}\delta(s_{i'}(t_{i'}),
e) \Right)^{\alpha_{j} - 1} \Right] \, .
\end{align*}
The last transition holds because the prior distribution $p$ is assumed to be a product distribution.

Now define a finite set of mutually independent Bernoulli random variables
$\lbrace X_{i', e}(s) \rbrace_{i' \neq i}$ such that each $X_{i', e}(s)$ takes
the value $1$ with probability $\sum_{t_{i'}: e \in
s_{i'}(t_{i'})}p_{i'}(t_{i'})$. Fixing a player $i'' \neq i$, we have
\begin{align*}
& \;
\mathbb{E}_{ \lbrace t_{i'} \sim p_{i'} \rbrace_{ i' \neq i} } \Left[ \Left(1 +
\sum_{i' \neq i}\delta(s_{i'}(t_{i'}), e) \Right)^{\alpha_{j} - 1} \Right] \\
=& \;
\sum_{t_{i''} \in T_{i''}}p_{i''}(t_{i''}) \mathbb{E}_{ \lbrace t_{i'} \sim
p_{i'} \rbrace_{ i' \neq i \wedge i' \neq i''} } \Left[ \Left(1 +
\delta(s_{i''}(t_{i''}), e) + \sum_{i' \neq i \wedge i' \neq
i''}\delta(s_{i'}(t_{i'}), e) \Right)^{\alpha_{j} - 1} \Right] \\
=& \;
\sum_{t_{i''}: e \in s_{i''}(t_{i''})}p_{i''}(t_{i''})\mathbb{E}_{ \lbrace
t_{i'} \sim p_{i'} \rbrace_{ i' \neq i \wedge i' \neq i''} } \Left[ \Left(1 + 1
+ \sum_{i' \neq i \wedge i' \neq i''}\delta(s_{i'}(t_{i'}),
e) \Right)^{\alpha_{j} - 1} \Right] \\
& \; +
\sum_{t_{i''}: e \notin s_{i''}(t_{i''})}p_{i''}(t_{i''})\mathbb{E}_{ \lbrace
t_{i'} \sim p_{i'} \rbrace_{ i' \neq i \wedge i' \neq i''} } \Left[ \Left(1 +
\sum_{i' \neq i \wedge i' \neq i''}\delta(s_{i'}(t_{i'}), e) \Right)^{\alpha_{j}
- 1} \Right] \\
=& \;
\mathbb{E}_{X_{i'', e}(s) }\Left[ \mathbb{E}_{ \lbrace t_{i'} \sim p_{i'}
\rbrace_{ i' \neq i \wedge i' \neq i''} } \Left[ \Left(1 + X_{i'', e}(s) +
\sum_{i' \neq i \wedge i' \neq i''}\delta(s_{i'}(t_{i'}), e)\Right)^{\alpha_{j}
- 1} \Right]  \Right] \, .
\end{align*} 
Therefore, it can be inductively proved that
\begin{equation*}
\mathbb{E}_{ \lbrace t_{i'} \sim p_{i'} \rbrace_{ i' \neq i} } \Left[ \Left(1
+ \sum_{i' \neq i}\delta(s_{i'}(t_{i'}), e)\Right)^{\alpha_{j} - 1} \Right] =
\mathbb{E} \Left[ \Left(1 + \sum_{i' \neq i}X_{i', e}(s) \Right)^{\alpha_{j} - 1}
\Right] \, .
\end{equation*}

Recall that the constant $1$ in the last expression above can also be viewed as a Bernoulli random variable which equals to $1$ with probability $1$. For every $\alpha_{j} \geq 2$, Lemma \ref{lemma_upper_bound_on_the_expectation_of_power_of_sum_of_random_variables} can be applied to obtain the following expression.
\begin{align*}
& \; \mathbb{E} \Left[ \Left(1 + \sum_{i' \neq i}X_{i', e}(s) \Right)^{\alpha_{j} - 1} \Right] \\
\leq& \; B_{\alpha_{j} - 1} \cdot \max\Left\lbrace \mathbb{E}\Left[ 1 + \sum_{i' \neq i}X_{i', e}(s) \Right], \Left(\mathbb{E}\Left[ 1 + \sum_{i' \neq i}X_{i', e}(s) \Right] \Right)^{\alpha_{j - 1}} \Right\rbrace\\
=& \; B_{\alpha_{j} - 1} \cdot \Left(\mathbb{E}\Left[ 1 + \sum_{i' \neq i}X_{i', e}(s) \Right] \Right)^{\alpha_{j} - 1}  \, .
\end{align*}
%
%
The second line holds because $\mathbb{E}\Big[ 1 + \sum_{i' \neq i}X_{i', e}(s) \Big] > 1$. Similarly, it can be derived from Lemma \ref{lemma_bounds_on_expectation_of_concave_monomial} that for every $\alpha_{j} \in (1, 2)$, 
\begin{equation*}
\mathbb{E} \bigg[ \Big(1 + \sum_{i' \neq i}X_{i', e}(s) \Big)^{\alpha_{j} - 1} \bigg] \; \leq \; \bigg(\mathbb{E}\Big[ 1 + \sum_{i' \neq i}X_{i', e}(s) \Big] \bigg)^{\alpha_{j} - 1} \, ,
\end{equation*}
which also trivially holds for $\alpha_{j} = 1$. 
So, $\mathbb{E} \bigg[ \Big(1 + \sum_{i' \neq i}X_{i', e}(s) \Big)^{\alpha_{j} - 1} \bigg] \leq \max\Big\lbrace 1, \, \max_{j: \alpha_{j} \geq 2} \InlineBreakForOverFull B_{\alpha_{j} - 1} \Big\rbrace \bigg(\mathbb{E}\Big[ 1 + \sum_{i' \neq i}X_{i', e}(s) \Big] \bigg)^{\alpha_{j} - 1}$, and in a similar way, it also be inferred from Lemma \ref{lemma_upper_bound_on_the_expectation_of_power_of_sum_of_random_variables} and Lemma \ref{lemma_bounds_on_expectation_of_concave_monomial} that $\mathbb{E} \bigg[ \Big(1 + \sum_{i' \neq i}X_{i', e}(s) \Big)^{\alpha_{j} - 1} \bigg] \geq \bigg(\mathbb{E}\Big[ 1 + \sum_{i' \neq i}X_{i', e}(s) \Big] \bigg)^{\alpha_{j} - 1} \bigg/ \max\Big\lbrace 1, \InlineBreakForOverFull \max_{ j: \alpha_{j} < 2}b_{\alpha_{j} - 1} \Big\rbrace$. 
%
%
Since $\mathbb{E}\Big[ 1 + \sum_{i' \neq i}X_{i', e}(s) \Big] = 1 + \sum_{i' \neq i}\sum_{t_{i'}}p_{i'}(t_{i'})\delta\Big(s_{i'}(t_{i'}), e\Big)$, Eq.~\eqref{formula_bounds_on_estimation_parameters_for_general_alpha_j} holds. 
For the special case where every $F_{e}$ is a quadratic function, by the linearity of the expectation, we have
\begin{equation*}
\mathbb{E}_{\lbrace{t_{i'} \sim p_{i'}\rbrace_{i' \neq i}}} \bigg[ \Big(1 + \sum_{i' \neq i}\delta(s_{i'}(t_{i'}), e)\Big)^{2 - 1} \bigg] \; = \; 1 + \sum_{i' \neq i}\mathbb{E}_{t_{i'} \sim p_{i'}}\Big[ \delta(s_{i'}(t_{i'}), e) \Big] \, ,
\end{equation*}
which completes the proof.
\end{proof}

\begin{corollary}\label{corollary_time_complexity_of_estimating_expected_individual_cost}
By computing Eq.~\eqref{formula_approximate_individual_cost}, the desired estimation of each expected cost share is obtained in $O(q \cdot \sum_{i \in [N]}|T_{i}|)$-time.
\end{corollary}

Plugging Theorem \ref{theorem_desired_smoothness_parameters_for_BGND_game}, Theorem \ref{theorem_existence_of_bounded_potential}, Theorem \ref{theorem_bounds_on_estimation_for_expected_individual_cost}, and Corollary \ref{corollary_time_complexity_of_estimating_expected_individual_cost} into Theorem \ref{theorem:approx-ratio-Alg-ABRD} proves our main result, Theorem \ref{theorem_main_result}.

\clearpage

\bibliographystyle{plainurl}
\bibliography{references}

\end{document}